\documentclass[12pt,a4paper]{article}

\usepackage{titlesec}
\titleformat*{\subsection}{\normalsize\bfseries}
\usepackage{tocloft}
\usepackage{authblk}
\usepackage{amsmath,amssymb,amscd,amsthm,latexsym,accents,stmaryrd}
\usepackage{mathrsfs}
\usepackage{mathtools}
\usepackage{amsfonts}
\usepackage{bbm}
\usepackage[latin9]{inputenc}
\usepackage[all,cmtip]{xy}
\usepackage{enumerate}
\usepackage{url}
\usepackage{hyperref}
\usepackage{tikz}
\usepackage{tikz-cd}
\usetikzlibrary{matrix,arrows,backgrounds}
\usetikzlibrary{cd}
\usetikzlibrary{decorations.text}
\usetikzlibrary{patterns}
\usetikzlibrary{automata,positioning}
\usepackage[parfill]{parskip}

\makeatletter
\def\thanks#1{\protected@xdef\@thanks{\@thanks
		\protect\footnotetext{#1}}}
\makeatother

\makeatletter
\newcommand{\subjclass}[2][2020]{%
	\let\@oldtitle\@title%
	\gdef\@title{\@oldtitle\footnotetext{#1 \emph{Mathematics subject classification:} #2}}%
}
\newcommand{\keywords}[1]{%
	\let\@@oldtitle\@title%
	\gdef\@title{\@@oldtitle\footnotetext{\emph{Key words and phrases:} #1.}}%
}
\makeatother

\makeatletter
\DeclareFontEncoding{LS2}{}{\@noaccents}
\makeatother
\DeclareFontSubstitution{LS2}{stix}{m}{n}
\DeclareSymbolFont{largesymbolsstix}{LS2}{stixex}{m}{n}
\DeclareMathDelimiter{\lbrbrak}{\mathopen}{largesymbolsstix}{"EE}{largesymbolsstix}{"14}
\DeclareMathDelimiter{\rbrbrak}{\mathclose}{largesymbolsstix}{"EF}{largesymbolsstix}{"15}

\makeatletter
\@namedef{subjclassname@2020}{%
	\textup{2020} Mathematics Subject Classification}
\makeatother

\newcommand{\nocontentsline}[3]{}
\newcommand{\tocless}[2]{\bgroup\let\addcontentsline=\nocontentsline#1{#2}\egroup}

\theoremstyle{definition}
\newtheorem{defin}{Definition}[section]
\newtheorem{rem}[defin]{Remark}

\theoremstyle{plain}
\newtheorem{theor}[defin]{Theorem}
\newtheorem{lem}[defin]{Lemma}
\newtheorem{prop}[defin]{Proposition}

\newtheoremstyle{dotless-thm}
{3pt}
{3pt}
{}
{}
{}
{.}
{.5em}
{}
\theoremstyle{dotless-thm}

\def\Z{{\mathbb{Z}}}

\def\A{{\mathcal{A}}}
\def\B{{\mathcal{B}}}
\def\C{{\mathcal{C}}}
\def\D{{\mathcal{D}}}
\def\G{{\mathcal{G}}}

\def\T{{\mathcal{T}}}

\def\P{{\mathcal{P}}}
\def\Q{{\mathcal{Q}}}
\def\R{{\mathbb{R}}}

\allowdisplaybreaks

\author[1]{Catarina Faustino}
\author[1]{Thomas Kahl}
\author[1]{Rodrigo Lopes}
\thanks{This research was supported by FCT (\emph{Funda\c c\~ao para a Ci\^encia e a Tecnologia}, Por\nolinebreak tugal) through the first author's PhD scholarship UI/BD/152071/2021 and projects UIDB/00013/2020 and UIDP/00013/2020.} 

\affil{\small{Centro de Matem\'atica,
		Universidade do Minho, 
		4710-057 Braga,
		Portugal\\
		\texttt{catarina.0109.cf@gmail.com}, 
		\texttt{kahl@math.uminho.pt}
		\texttt{rodrigorok2.ML@gmail.com}
	}
}

\begin{document}

\title{On the topology of concurrent systems}

\date{}



\maketitle

\begin{abstract}
	Higher-dimensional automata, i.e., pointed labeled precubical sets, are a powerful combinatorial-topological model for concurrent systems. In this paper, we show that for every (nonempty) connected polyhedron there exists a shared-variable system such that the higher-dimensional automaton modeling the state space of the system has the homotopy type of the polyhedron.
\end{abstract}

\section{Introduction}
\begin{sloppypar}
As amply demonstrated in the literature, concepts and methods from algebraic topology can be profitably employed in concurrency theory, the field of computer science that studies systems of simultaneously executing processes, see, e.g., \cite{Goubault, FajstrupGR, vanGlabbeek, FGHMR}. Several topological models of concurrency have been introduced by various authors, e.g., \cite{FajstrupGR, GrandisBook, Krishnan}. A particularly important combinatorial-topological model of concurrency is given by higher-dimensional automata \cite{vanGlabbeek}, which go back to \cite{Pratt}. It has been shown in \cite{vanGlabbeek} that higher-dimensional automata are more expressive than the principal traditional models of concurrency. 
\end{sloppypar}

A higher-dimensional automaton (HDA) is a pointed precubical set (cubical set without degeneracy maps) with edge labeling such that opposite edges of \(\mbox{2-}\)cubes have the same label. The vertices of an HDA represent the states of a concurrent system, with the base vertex corresponding to the initial state. The labeled edges model the transitions of the system, and two- and higher-dimensional cubes express the independence of transitions: an \(n\)-cube in an HDA indicates that the \(n\) transitions starting from its origin are independent in the sense that they can occur in any order or even simultaneously without any observable difference.  

\begin{sloppypar}
	A standard procedure for constructing an HDA model of a concurrent system is to first construct a transition system and then fill in empty squares and higher-dimensional cubes, see,  e.g., \cite{vanGlabbeek, GaucherCombinatorics, GoubaultMimram, transhda}. To make this more precise, consider the example of Peterson's algorithm, a protocol designed to give two processes fair and mutually exclusive access to a shared resource \cite{Peterson}. Peterson's algorithm is based on three shared variables---namely, a variable \({t}\) whose possible values are the process IDs, say \({0}\) and \({1}\), and two  boolean variables \({b_0}\) and \({b_1}\). Process \({i}\) executes the following protocol with four local states and four transitions: 
	\begin{itemize}
		\item Set \({b_i}\) to \({1}\) to indicate the intention to enter the ``critical section".
		\item Set \({t}\) to \({1-i}\) to give priority to the other process.
		\item Wait until \({b_{1-i} = 0}\) or \({t = i}\), and then enter the critical section. 
		\item Leave the critical section setting \({b_i}\) to \({0}\), and repeat the procedure from the beginning. 
	\end{itemize}
    To start, all variables are set to \( 0\). A global state of such a shared-variable system is a tuple whose components are local states of the processes and values of the variables. The transition system associated with the shared-variable system is a labeled directed graph whose vertices correspond to  the global states that are actually visited during some execution of the system, and whose edges model the transitions between these global states. 
    The transitions starting from a given global state correspond to the actions that are enabled in that state. These actions are specified in the edge labels, indexed by the respective process IDs. The HDA model of the system is then constructed from the transition system as a kind of coskeleton, i.e., by suitably filling in empty cubes of dimensions \(\geq 2\), see Figure \ref{peterHDA} for the case of Peterson's algorithm.
	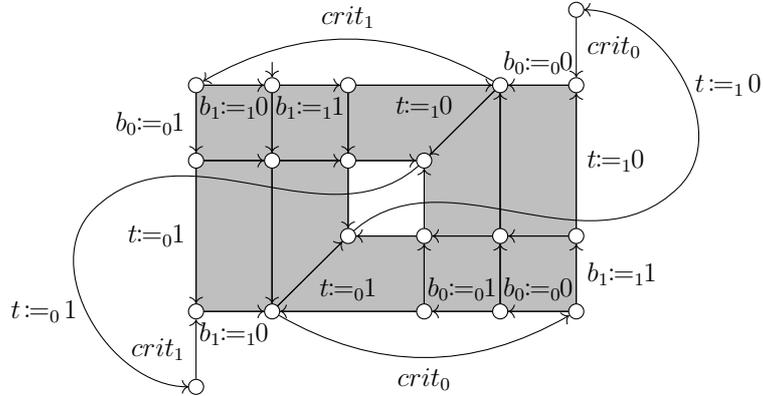
\begin{figure}[t]
		\center
		\begin{tikzpicture}[initial text={},on grid]

			\path[draw, fill=lightgray] (0,0)--(1,0)--(1,-1)--(0,-1)--cycle;    
			
			\path[draw, fill=lightgray] (1,0)--(2,0)--(2,-1)--(1,-1)--cycle;
			
			\path[draw, fill=lightgray] (2,0)--(4,0)--(3,-1)--(2,-1)--cycle;
			
			\path[draw, fill=lightgray] (0,-1)--(1,-1)--(1,-3)--(0,-3)--cycle;
			
			\path[draw, fill=lightgray] (1,-1)--(2,-1)--(2,-2)--(1,-3)--cycle;
			
			\path[draw, fill=lightgray] (3,-1)--(4,0)--(4,-2)--(3,-2)--cycle;
			
			\path[draw, fill=lightgray] (4,0)--(5,0)--(5,-2)--(4,-2)--cycle;
			
			\path[draw, fill=lightgray] (2,-2)--(3,-2)--(3,-3)--(1,-3)--cycle;
			
			\path[draw, fill=lightgray] (3,-2)--(4,-2)--(4,-3)--(3,-3)--cycle;
			
			\path[draw, fill=lightgray] (4,-2)--(5,-2)--(5,-3)--(4,-3)--cycle;
			
			\draw[] (3,-1) to [out=225,in=45] (-1.35,-1.65); 
			
			\draw[->] (-1.35,-1.65) to [out=225,in=180] (-0.1,-4);
			
			\node at (-2,-3) {\scalebox{0.85}{${t\!\coloneqq_0\! 1}$}};
			
			\draw[] (2,-2) to [out=45,in=225] (6.35,-1.35); 
			
			\draw[->] (6.35,-1.35) to [out=45,in=0] (5.1,1);

			\node at (7,0) {\scalebox{0.85}{${t\!\coloneqq_1\! 0}$}};

			\node[state,minimum size=0pt,inner sep =2pt,fill=white] (q_0) at (0,0)  {}; 
			
			\node[state,minimum size=0pt,inner sep =2pt,fill=white, initial,initial where=above,initial distance=0.2cm] (q_1) [right=of q_0,xshift=0cm] {};
			
			\node[state,minimum size=0pt,inner sep =2pt,fill=white] (q_2) [right=of q_1,xshift=0cm] {};
			
			\node[state,minimum size=0pt,inner sep =2pt,fill=white] (q_3) [right=of q_2,xshift=1cm] {};
			
			\node[state,minimum size=0pt,inner sep =2pt,fill=white] (q_4) [right=of q_3,xshift=0cm] {};
			
			\node[state,minimum size=0pt,inner sep =2pt,fill=white] (q_5) [below=of q_0,xshift=0cm] {};  
			
			\node[state,minimum size=0pt,inner sep =2pt,fill=white] (q_6) [right=of q_5,xshift=0cm] {};
			
			\node[state,minimum size=0pt,inner sep =2pt,fill=white] (q_7) [right=of q_6,xshift=0cm] {};
			
			\node[state,minimum size=0pt,inner sep =2pt,fill=white] (q_8) [right=of q_7,xshift=0cm] {};
			
			\node[state,minimum size=0pt,inner sep =2pt,fill=white] (q_9) [below=of q_7,xshift=0cm] {};
			
			\node[state,minimum size=0pt,inner sep =2pt,fill=white] (q_10) [right=of q_9,xshift=0cm] {};
			
			\node[state,minimum size=0pt,inner sep =2pt,fill=white] (q_11) [right=of q_10,xshift=0cm] {};
			
			\node[state,minimum size=0pt,inner sep =2pt,fill=white] (q_12) [right=of q_11,xshift=0cm] {};
			
			\node[state,minimum size=0pt,inner sep =2pt,fill=white] (q_17) [below=of q_12,xshift=0cm] {};
			
			\node[state,minimum size=0pt,inner sep =2pt,fill=white] (q_16) [left=of q_17,xshift=0cm] {};
			
			\node[state,minimum size=0pt,inner sep =2pt,fill=white] (q_15) [left=of q_16,xshift=0cm] {};
			
			\node[state,minimum size=0pt,inner sep =2pt,fill=white] (q_14) [left=of q_15,xshift=-1cm] {};
			
			\node[state,minimum size=0pt,inner sep =2pt,fill=white] (q_13) [left=of q_14,xshift=0cm] {};
			
			\node[state,minimum size=0pt,inner sep =2pt,fill=white] (p_0) [above=of q_4,xshift=0cm] {};
			
			\node[state,minimum size=0pt,inner sep =2pt,fill=white] (p_1) [below=of q_13,xshift=0cm] {};

			\path[->] 
			(p_0) edge[right] node[] {\scalebox{0.85}{${crit_0}$}} (q_4)
			
			(q_0) edge[below] node[] {\scalebox{0.85}{${b_1\!\!\coloneqq_1\!\! 0}$}} (q_1)
			(q_1) edge[below] node[] {\scalebox{0.85}{${b_1\!\!\coloneqq_1\!\! 1}$}} (q_2)
			(q_2) edge[below] node[] {\scalebox{0.85}{${t\!\!\coloneqq_1\!\! 0}$}} (q_3)
			(q_4) edge[above] node[] {\scalebox{0.85}{${b_0\!\!\coloneqq_0\!\! 0}$}} (q_3)
			
			(q_0) edge[left,] node[] {\scalebox{0.85}{${b_0\!\!\coloneqq_0\!\! 1}$}} (q_5)
			(q_1) edge[above] node {} (q_6)
			(q_2) edge[above] node {} (q_7)
			(q_3) edge[above] node {} (q_8)
			
			(q_5) edge[above] node {} (q_6)
			(q_6) edge[above] node {} (q_7)
			(q_7) edge[above] node {} (q_8)
			
			(q_7) edge[above] node {} (q_9)
			(q_10) edge[above] node {} (q_8)
			(q_11) edge[above] node {} (q_3)
			(q_12) edge[right] node[] {\scalebox{0.85}{${t\!\!\coloneqq_1\!\! 0}$}} (q_4)
			
			(q_10) edge[above] node {} (q_9)
			(q_11) edge[above] node {} (q_10)
			(q_12) edge[above] node {} (q_11)
			
			(q_5) edge[left] node[] {\scalebox{0.85}{${t\!\!\coloneqq_0\!\! 1}$}} (q_13)
			(q_6) edge[above] node {} (q_14)
			(q_14) edge[above] node {} (q_9)
			(q_15) edge[above] node {} (q_10)
			(q_16) edge[above] node {} (q_11)
			(q_17) edge[right] node[] {\scalebox{0.85}{${b_1\!\!\coloneqq_1\!\! 1}$}} (q_12)
			
			(q_13) edge[below] node[] {\scalebox{0.85}{${b_1\!\!\coloneqq_1\!\! 0}$}} (q_14)
			(q_15) edge[above] node[] {\scalebox{0.85}{${t\!\!\coloneqq_0\!\! 1}$}} (q_14)
			(q_16) edge[above] node[] {\scalebox{0.85}{${b_0\!\!\coloneqq_0\!\! 1}$}} (q_15)
			(q_17) edge[above] node[] {\scalebox{0.85}{${b_0\!\!\coloneqq_0\!\! 0}$}} (q_16)
			
			(p_1) edge[left] node[] {\scalebox{0.85}{${crit_1}$}} (q_13)
			
			(q_3) edge[above, bend right] node[] {\scalebox{0.85}{${crit_1}$}} (q_0)
			(q_14) edge[below, bend right] node[] {\scalebox{0.85}{${crit_0}$}} (q_17)
			;

		\end{tikzpicture}
		\caption{HDA for Peterson's algorithm. Parallel arrows are supposed to have the same label, and the small arrow indicates the initial state.}\label{peterHDA}
	\end{figure}
\end{sloppypar}

\begin{sloppypar}
    It turns out that the topological analysis of HDAs provides information that is relevant from the point of view of computer science. Indeed, two executions of a concurrent system can be considered equivalent if and only if they can be modeled as directed paths that are homotopic in a directed sense, see, e.g., \cite{Goubault, FGHMR}. Additionally, the homology of an HDA model of a concurrent system contains global information about the independence of processes and components of the system \cite{labels, wehda}. Further connections between algebraic topology and concurrency theory are developed in \cite{Goubault, FajstrupGR, FGHMR}.
\end{sloppypar}

\iffalse 
It turns out that the topological analysis of HDAs provides information that is relevant from the point of view of computer science. Indeed, two executions of a concurrent system can be considered equivalent if and only if they are modeled by directed paths that are homotopic in a directed sense, see, e.g., \cite{Goubault, FGHMR}. Additionally, the homology of an HDA contains global information about the independence of processes and components of the modeled concurrent system \cite{labels, wehda}. Further connections between algebraic topology and concurrency theory are developed in \cite{Goubault, FajstrupGR, FGHMR}.
\fi

The purpose of this paper is to show that the topology of an HDA model of a concurrent system can be arbitrarily complex. More precisely, we show that for every (nonempty) connected polyhedron there exists a shared-variable system whose HDA model has the same homotopy type as the polyhedron. This is similar in spirit to \cite{ZiemianskiOnExecutionSpaces}, where it is shown that for every connected polyhedron there exists a PV-program (a particular kind of shared-variable system) whose execution space contains a connected component with the same homotopy type as the polyhedron. This paper is also related to \cite{FajstrupCubicalLocal}, where in particular it is shown that every polyhedron admits a cubical local partial order. In fact, the first step in the proof of our result is to show that the cubical barycentric subdivision of a simplicial complex can be constructed as a precubical set. This actually strengthens \cite[Cor. 3.13]{FajstrupCubicalLocal} because it shows that no further subdivision of the cubical barycentric subdivision is needed to equip a polyhedron with a cubical local partial order. 

The paper is organized as follows. The precubical set corresponding to the cubical barycentric subdivision of a simplicial complex is constructed in Section \ref{SecK}. In Section \ref{SecP}, we turn this precubical set into an HDA, which we show to be an HDA model of the transition system given by its 1-skeleton. In the next section, we show that one can replace this HDA by a homotopy equivalent accessible one, i.e., an HDA where all states are reachable by a directed path from the initial state. In the last section, we then show that this accessible HDA is isomorphic to the HDA model of a shared-variable system.

\section{The simplicial complex \texorpdfstring{\(K\)}{} and its cubical barycentric subdivision \texorpdfstring{\(P\)}{}} \label{SecK}

Throughout this paper, we consider a connected abstract simplicial complex \(K\) with vertices \(1, \dots, N\) and the associated polyhedron \(|K|\), which we view as a subspace of the standard \((N-1)\)-simplex \(\Delta^{N-1} \subseteq \R^{N}\). More precisely, we define \(|K|\) to be the subspace \(\bigcup \limits_{\sigma \in K} |\sigma| \subseteq \Delta^{N-1}\), where for a simplex \(\sigma \in K\), \(|\sigma|\) is the geometric simplex in \(\R^{N}\) with vertices \(e_i = (0, \dots, 0, \underset{i}{1}, 0, \dots 0)\), \(i \in \sigma\). In this section, we construct the cubical barycentric subdivision of \(K\) as a precubical set.

\subsection*{Precubical sets}

Let us briefly recall some basic concepts about precubical sets. A \emph{precubical set} is a graded set \(X = (X_n)_{n \geq 0}\) with \emph{face maps} \(d^k_i\colon X_n \to X_{n-1}\) \(({n>0,}\;{k\in\{0,1\}},\; {i \in \{1, \dots, n\}})\) satisfying the relations \(d^k_i\circ d^l_{j}= d^l_{j-1}\circ d^k_i\) \(({k,l \in \{0,1\},\; i<j})\). If \(x\in X_n\), we say that \(x\) is of  \emph{degree} or \emph{dimension} \(n\). The elements of degree \(n\) are called the \emph{\(n\)-cubes} of \(X\). The elements of degree \(0\) are also called the \emph{vertices} of \(X\), and the \(1\)-cubes are also called the \emph{edges} of \(X\). Precubical sets form a category in which  morphisms are morphisms of graded sets that are compatible with the face maps. A \emph{precubical subset} of a precubical set \(X\) is a graded subset of \(X\) that is stable under the face maps. The \emph{tensor product} of two precubical sets \(X\) and \(Y\) is the precubical set \(X\otimes Y\) given by \[(X\otimes Y)_n = \coprod \limits_{p+q = n} X_p\times Y_q \] 
and
\[d_i^k(x,y) = \left\{ \begin{array}{ll} (d_i^kx,y), & 1\leq i \leq p,\\
	(x,d_{i-p}^ky), & p < i \leq n,
\end{array}\right. \quad (x,y) \in X_p\times Y_q.\] The \emph{geometric realization} of a precubical set \(X\) is the quotient space \[|X|=\left(\coprod _{n \geq 0} X_n \times [0,1]^n\right)/\sim\] where the sets \(X_n\) are given the discrete topology and the equivalence relation is generated by
\[(d^k_ix,u) \sim (x,\delta_i^k(u)), \quad  x \in X_{n+1},\; u\in [0,1]^n,\; i \in  \{1, \dots, n+1\},\; k \in \{0,1\}.\]
Here the map \(\delta^k_i\colon [0,1]^n \to [0,1]^{n+1}\) is defined by \[\delta^k_i(u_1, \dots, u_n) = (u_1, \dots, u_{i-1},k,u_i, \dots, u_n).\]
The geometric realization is functorial. For a morphism of precubical sets \(f\colon X \to Y\), the continuous map \(|f|\colon |X| \to |Y|\) is given by \(|f|([x,u])= [f(x),u]\). The geometric realization of a precubical set is a CW complex. The \(n\mbox{-}\)skeleton of \(|X|\) is the geometric realization of the precubical subset \(X_{\leq n}\) of \(X\) defined by \((X_{\leq n})_i = X_i\) for \(i \leq n\) and \((X_{\leq n})_i = \emptyset\) for \(i > n\).

\subsection*{The precubical set \(P\)}

The \emph{cubical barycentric subdivision} of \(K\) is the precubical set \(P\) where the elements of \(P_n\) are pairs \((\tau, \sigma)\) of simplexes of \(K\) such that \(\tau\) is a face of \(\sigma\) and \(\sigma \setminus \tau\) has \(n\) elements, see Figure \ref{cube} for a picture. 
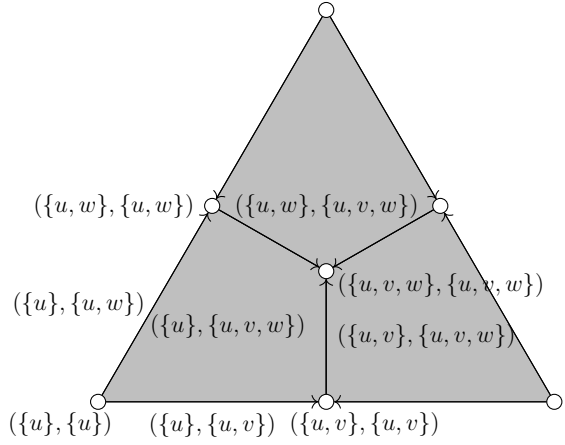
\begin{figure}[t]
	\center
	\begin{tikzpicture}[initial text={},on grid, scale=1]

		\path[draw, fill=lightgray] (0,0)--(3,0)--(3,1.732051)--(1.5,2.598076)--cycle;    
		\path[draw, fill=lightgray]
		(3,0)--(3,1.732051)--(4.5,2.598076)--(6,0)--cycle;
		\path[draw, fill=lightgray]
		(3,1.732051)--(4.5,2.598076)--(3,5.196152)--(1.5,2.598076)--cycle;
		
		\node[state,minimum size=0pt,inner sep =2pt,fill=white] (q_0) at (0,0)  {};
		\node[state,minimum size=0pt,inner sep =2pt,fill=white] (q_1) at (3,0)  {};
		\node[state,minimum size=0pt,inner sep =2pt,fill=white] (q_2) at (6,0)  {};
		\node[state,minimum size=0pt,inner sep =2pt,fill=white] (q_3) at (1.5,2.598076)  {};
		\node[state,minimum size=0pt,inner sep =2pt,fill=white] (q_4) at (3,1.732051)  {};
		\node[state,minimum size=0pt,inner sep =2pt,fill=white] (q_5) at (4.5,2.598076)  {};
		\node[state,minimum size=0pt,inner sep =2pt,fill=white] (q_6) at (3,5.196152)  {};
		
		\path[->] 
		(q_0) edge[below] node[] {\scalebox{0.75}{$(\{u\},\{u,v\})$}} (q_1)
		(q_0) edge[left] node[] {\scalebox{0.75}{$(\{u\},\{u,w\})$}} (q_3)
		(q_1) edge[right] node[] {\scalebox{0.75}{$(\{u,v\},\{u,v,w\})$}} (q_4)
		(q_3) edge[above right] node[] {} (q_4)
		(q_5) edge[above left] node[] {} (q_4)
		(q_2) edge[above right] node[] {} (q_1)
		(q_2) edge[above right] node[] {} (q_5)
		(q_6) edge[above right] node[] {} (q_3)
		(q_6) edge[above right] node[] {} (q_5);
		
		\node at (1.7,1) {\scalebox{0.75}{$(\{u\},\{u,v,w\})$}};
		\node at (-0.5,-0.3) {\scalebox{0.75}{$(\{u\},\{u\})$}};
		\node at (3.5,-0.3) {\scalebox{0.75}{$(\{u,v\},\{u,v\})$}};
		\node at (4.5,1.55) {\scalebox{0.75}{$(\{u,v,w\},\{u,v,w\})$}};
		\node at (0.2,2.598076) {\scalebox{0.75}{$(\{u,w\},\{u,w\})$}};
		\node at (3.0,2.6)
		{\scalebox{0.75}{$(\{u,w\},\{u,v,w\})$}};
	\end{tikzpicture}
	\caption{A 2-cube \((\{u\},\{u,v,w\})\) and its faces}\label{cube}
\end{figure} 
Considering the natural order on the set of vertices of \(K\), the face maps of \(P\) are defined as follows: if \(\sigma \setminus \tau = \{w_1 < \cdots < w_n\}\), we set 
\[d^0_i(\tau, \sigma) = (\tau, \sigma\setminus\{w_i\} )\] and \[d^1_i(\tau, \sigma) = (\tau\cup \{w_i\}, \sigma).\] 
One easily checks that \(P\) is indeed a precubical set. The remainder of this section is devoted to the proof that \(|P|\approx |K|\).

\subsection*{The map \texorpdfstring{\(f\colon |P| \to |K|\)}{}}

Let \((\tau, \sigma) \in P_n\), and suppose that \[\tau = \{v_1 < \cdots < v_r\}\quad \mbox{and} \quad  \sigma \setminus \tau = \{w_1 < \cdots < w_n\}.\] 
We decompose the standard \(n\)-cube \([0,1]^n\) as the union of the \(n\)-simplexes
\[\Delta_\theta = \{(t_1, \dots, t_n) \in [0,1]^n\mid  t_{\theta(1)} \geq \cdots \geq t_{\theta(n)}\}, \quad \theta \in S_n\]
and define a continuous map \(f_{\tau, \sigma} \colon [0,1]^n \to |K|\) by setting for an element \((t_1, \dots , t_n) \in \Delta_\theta\), 
\begin{align*}
	\MoveEqLeft{f_{\tau, \sigma}(t_1, \dots , t_n)}\\ &= \tfrac{1 -t_{\theta(1)}}{r}\sum \limits_{i=1}^r e_{v_i} + \tfrac{t_{\theta(1)} -t_{\theta(2)}}{r+1}(\sum \limits_{i=1}^r e_{v_i} + e_{w_{\theta(1)}}) + \cdots \\
	&\quad  +  \tfrac{t_{\theta(n-1)} -t_{\theta(n)}}{r+n-1}(\sum \limits_{i=1}^r e_{v_i} + \sum \limits_{i=1}^{n-1} e_{w_{\theta(i)}}) + \tfrac{t_{\theta(n)}}{r+n}(\sum \limits_{i=1}^r e_{v_i} + \sum \limits_{i=1}^n e_{w_{\theta(i)}})
\end{align*}
(see Figure \ref{f}). If \(\sigma = \tau\), this formula is to be interpreted in such a way that \(f_{\tau, \sigma}(()) = \frac{1}{r} \sum \limits_{i=1}^r e_{v_i}\). Note that \(f_{\tau, \sigma}(t_1, \dots , t_n)\) is a convex combination of barycenters of faces of \(|\sigma|\) and hence itself an element of \(|\sigma|\). By  the following fact, the proof of which is left to the reader, \(f_{\tau, \sigma}\) is well defined:

\begin{lem} \label{intersect}
	Let \(\theta, \psi \in S_n\) and \((t_1, \dots, t_n) \in \Delta_\theta\cap \Delta_\psi\). Then \(t_{\theta(i)} = t_{\psi(i)}\) for all \(i \in \{1, \dots, n\}\). Moreover, if \(1 \leq i_1 < \dots < i_k = n\) are indices such that 
	\[t_{\theta(1)} = \dots = t_{\theta(i_1)} >  t_{\theta(i_1+1)} = \dots > t_{\theta(i_{k-1}+1)} = \dots = t_{\theta(n)},\]
	then \(\{\theta(1), \dots , \theta(i_s)\} = \{\psi(1), \dots, \psi(i_s)\}\) for all \(s \in \{1, \dots, k\}\).	
\end{lem}

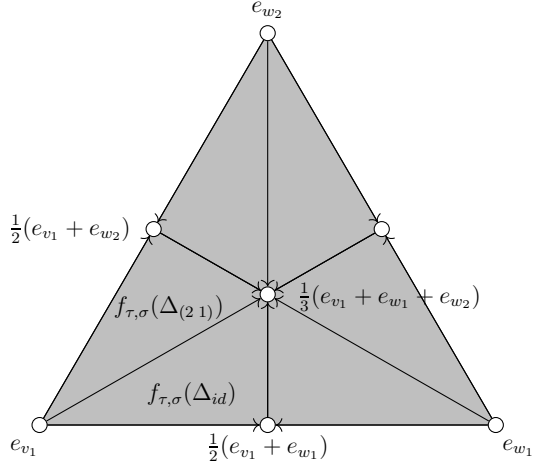
\begin{figure}[t]
	\center
	\begin{tikzpicture}[initial text={},on grid, scale=1]

		\path[draw, fill=lightgray] (0,0)--(3,0)--(3,1.732051)--(1.5,2.598076)--cycle;    
		\path[draw, fill=lightgray]
		(3,0)--(3,1.732051)--(4.5,2.598076)--(6,0)--cycle;
		\path[draw, fill=lightgray]
		(3,1.732051)--(4.5,2.598076)--(3,5.196152)--(1.5,2.598076)--cycle;
		
		\node[state,minimum size=0pt,inner sep =2pt,fill=white] (q_0) at (0,0)  {};
		\node[state,minimum size=0pt,inner sep =2pt,fill=white] (q_1) at (3,0)  {};
		\node[state,minimum size=0pt,inner sep =2pt,fill=white] (q_2) at (6,0)  {};
		\node[state,minimum size=0pt,inner sep =2pt,fill=white] (q_3) at (1.5,2.598076)  {};
		\node[state,minimum size=0pt,inner sep =2pt,fill=white] (q_4) at (3,1.732051)  {};
		\node[state,minimum size=0pt,inner sep =2pt,fill=white] (q_5) at (4.5,2.598076)  {};
		\node[state,minimum size=0pt,inner sep =2pt,fill=white] (q_6) at (3,5.196152)  {};
		
		\path[->] 
		(q_0) edge[below] node[] {\scalebox{0.75}{}} (q_1)
		(q_0) edge[left] node[] {\scalebox{0.75}{}} (q_3)
		(q_1) edge[right] node[] {\scalebox{0.75}{}} (q_4)
		(q_3) edge[above right] node[] {} (q_4)
		(q_5) edge[above left] node[] {\scalebox{0.75}{}} (q_4)
		(q_2) edge[above right] node[] {} (q_1)
		(q_2) edge[above right] node[] {} (q_5)
		(q_6) edge[above right] node[] {} (q_3)
		(q_6) edge[above right] node[] {} (q_5)
		(q_0) edge[above right] node[] {} (q_4)
		(q_2) edge[above right] node[] {} (q_4)
		(q_6) edge[above right] node[] {} (q_4);
		
		\node at (1.7,1.55) {\scalebox{0.75}{$f_{\tau,\sigma}(\Delta_{(2\; 1)})$}};
		\node at (2,0.4) {\scalebox{0.75}{$f_{\tau,\sigma}(\Delta_{id})$}};
		\node at (-0.2,-0.3) {\scalebox{0.75}{$e_{v_1}$}};
		\node at (6.3,-0.3) {\scalebox{0.75}{$e_{w_1} $}};
		\node at (3,5.5) {\scalebox{0.75}{$e_{w_2}$}};
		\node at (3,-0.3) {\scalebox{0.75}{$\frac{1}{2}(e_{v_1}+e_{w_1})$}};
		\node at (4.6,1.7) {\scalebox{0.75}{$\frac{1}{3}(e_{v_1}+e_{w_1}+e_{w_2})$}};
		\node at (0.4,2.6) {\scalebox{0.75}{$\frac{1}{2}(e_{v_1}+e_{w_2})$}};
				
	\end{tikzpicture}
	\caption{The image of the map \(f_{\tau, \sigma}\) for \(\tau =\{v_1\}\) and \(\sigma \setminus \tau = \{w_1<w_2\}\)}\label{f}
\end{figure}

The next lemma shows that a well-defined continuous map \(f \colon |P| \to |K|\) is given by \[f([(\tau, \sigma),(t_1, \dots, t_n)]) = f_{\tau, \sigma}(t_1, \dots, t_n), \quad (\tau, \sigma) \in P_n,\, (t_1, \dots, t_n) \in [0,1]^n.\]

\begin{lem}
	Let \((\tau, \sigma) \in P_n\) \((n \geq 1)\), and suppose that \(\tau = \{v_1 < \cdots < v_r\}\) and \(\sigma \setminus \tau = \{w_1 < \cdots < w_{n}\}\). Then 
	\(f_{\tau, \sigma\setminus \{w_i\}} = f_{\tau, \sigma} \circ \delta^0_i \colon [0,1]^{n-1} \to |K|\) and \(f_{\tau \cup \{w_i\}, \sigma} = f_{\tau, \sigma} \circ \delta^1_i \colon [0,1]^{n-1} \to |K|\). 
\end{lem}

\begin{proof}
	We prove only the first equality and leave the similar proof of the second to the reader. Let us first note that 
	\begin{align*}
		( \sigma \setminus \{ w_i \} ) \setminus \tau & =  \{ w_1 < \dots < w_{i-1} < w_{i+1} < \dots < w_{n} \} \\ & =  \{ \bar w_1 < \dots < \bar w_{i-1} < \bar w_i < \dots < \bar w_{n-1}\}
	\end{align*} 	
	with 	
	\begin{align*}
		\bar w_j =
		\begin{cases}
			w_j, &  j<i, \\
			w_{j+1}, &  j\geq i.
		\end{cases}
	\end{align*}
	Let $\theta \in S_{n-1}$ and $(t_1, \dots, t_{n-1}) \in \Delta_\theta \subseteq  [0,1]^{n-1}$. Then $t_{\theta(1)} \geq \dots \geq t_{\theta(n-1)}$, and so defining $(x_1, \dots, x_n)  = (t_1, \dots, t_{i-1}, 0, t_i, \dots, t_{n-1})$, we will have that $ (x_1, \dots, x_n) \in \Delta_\psi \subseteq [0,1]^n$, with $\psi \in S_n$ defined by  
	\begin{align*}
		\psi (j) =
		\begin{cases}
			\theta(j), &  j<n, \medspace \theta(j)<i, \\
			\theta(j) + 1, &  j<n, \medspace \theta(j)\geq i, \\
			i, & j=n.
		\end{cases}
	\end{align*}
	Indeed, 
	\begin{align*}
		x_{\psi (j)} &=
		\begin{cases}
			x_{\theta(j)} = t_{\theta(j)}, &  j<n, \medspace \theta(j)<i, \\
			x_{\theta(j) + 1} = t_{\theta(j)}, &   j<n, \medspace \theta(j)\geq i, \\
			x_i = 0, & j=n
		\end{cases}  \quad= 
		\begin{cases}
			t_{\theta(j)}, &  j<n, \\
			0, & j=n.
		\end{cases} 
	\end{align*}
	So $x_{\psi(1)} \geq \dots \geq x_{\psi(n)}$, i.e.,  $(x_1, \dots, x_n) \in \Delta_\psi$. Note also that 
	\begin{align*}
		w_{\psi (j)} & =
		\begin{cases}
			w_{\theta(j)} = \bar w_{\theta(j)}, &  j<n, \medspace \theta(j)<i, \\
			w_{\theta(j) + 1} = \bar w_{\theta(j)}, & j<n, \medspace \theta(j)\geq i, \\
			w_i, & j=n
		\end{cases} \quad =
		\begin{cases}
			\bar w_{\theta(j)}, & j<n, \\
			w_i, & j=n.
		\end{cases} 
	\end{align*}
	We finally have that
	\begin{align*}
		\MoveEqLeft{f_{\tau, \sigma} \circ \delta_i^0 (t_1 \dots, t_{n-1})} = 	f_{\tau, \sigma} (t_1 \dots, t_{i-1}, 0, t_i, \dots t_{n-1}) = f_{\tau, \sigma} (x_1, \dots, x_{n}) \\  
		& = \tfrac{1-x_{\psi (1)}}{r} \sum_{l=1}^{r} e_{v_l} + \tfrac{x_{\psi (1)} - x_{\psi(2)}}{r+1} (\sum_{l=1}^{r} e_{v_l} + e_{w_{\psi(1)}})+ \dots \\ &\quad   + \tfrac{x_{\psi (n-1)} - x_{\psi(n)}}{r+n-1} (\sum_{l=1}^{r} e_{v_l} + \sum_{l=1}^{n-1} e_{w_{\psi(l)}})   + \tfrac{x_{\psi (n)}}{r+n} (\sum_{l=1}^{r} e_{v_l} + \sum_{l=1}^{n} e_{w_{\psi(l)}}) \\ 
		& =  \tfrac{1-t_{\theta (1)}}{r} \sum_{l=1}^{r} e_{v_l} + \tfrac{t_{\theta (1)} - t_{\theta(2)}}{r+1} (\sum_{l=1}^{r} e_{v_l} + e_{\bar w_{\theta(1)}}) + \dots\\ & \quad  + \tfrac{t_{\theta (n-1)} - 0}{r+n-1} (\sum_{l=1}^{r} e_{v_l} + \sum_{l=1}^{n-1} e_{\bar w_{\theta(l)}})   + \tfrac{0}{r+n} (\sum_{l=1}^{r} e_{v_l} + \sum_{l=1}^{n-1} e_{\bar w_{\theta(l)}} + e_{w_i}) \\		
		& = f_{\tau, \sigma \setminus \{ w_i \}} (t_1, \dots, t_{n-1}). \qedhere
	\end{align*}	
\end{proof}

We can now prove the main result of this section:

\begin{theor} \label{homeo}
	The map \(f \colon |P| \to |K|\) is a homeomorphism.
\end{theor}

\begin{proof}
	In order to define a map \(g\colon |K| \to |P|\), consider an element \(x\in |K|\). Let \(\sigma = {\{u_1 < \cdots < u_n\}}\) be the unique simplex of \(K\) such that \(x = \sum \limits_{i = 1}^n s_ie_{u_i}\) for (unique) numbers \(s_i > 0\) such that \(\sum \limits_{i = 1}^n s_i = 1\). Choose a permutation \(\alpha \in S_n\) such that \(s_{\alpha(1)} \geq \cdots \geq s_{\alpha(n)}\). 
	Let \[m = \max\{i \in \{1, \dots, n\}\,|\, s_{\alpha(i)} = s_{\alpha (1)}\},\] and set \(\tau = \{u_{\alpha(1)}, \dots, u_{\alpha(m)}\}\). Let \(\phi\) be the unique order-isomorphism \[\{\alpha(m+1), \dots, \alpha(n)\} \to \{1, \dots, n-m\},\] and define \(\theta \in S_{n-m}\) by
	\(\theta(i) = \phi(\alpha(m+i))\). Set \[t_{\theta(i)} = (m+i)s_{\alpha(m+i)} + \sum \limits_{j = m+i+1}^ns_{\alpha(j)}.\] Then \[0 \leq t_{\theta(n-m)} \leq \cdots \leq t_{\theta(1)} \leq \sum \limits_{j = 1}^n s_{\alpha(j)} = 1.\] Hence \({(t_1, \dots, t_{n-m}) \in \Delta_{\theta}} \subseteq [0,1]^{n-m}\). We set \[g(x) = [(\tau,\sigma),(t_1, \dots, t_{n-m})].\]
	Using Lemma \ref{intersect}, one easily checks that \(g(x)\) does not depend on the choice of the permutation \(\alpha\). We have thus defined a map \(g\colon |K| \to |P|\). Tedious but rather straightforward computations now show that \(f\circ g = id_{|K|}\) and \(g\circ f = id_{|P|}\). Since \(f\) is a continuous map between compact Hausdorff spaces, it follows that \(f\) is a homeomorphism.
\end{proof}

\section{The HDA \texorpdfstring{\(\P\)}{}} \label{SecP}

A \emph{higher-dimensional automaton} (HDA) is a tuple \(\A = (P_{\A},I_{\A},\Sigma_{\A}, \lambda_{\A})\) where  $P_{\A}$ is a precubical set, ${I_{\A} \in (P_{\A})_0}$ is a vertex, called the \emph{initial state}, $\Sigma_{\A}$ is a finite set of \emph{labels}, and $\lambda_\A \colon (P_{\A})_1 \to \Sigma_{\A}$ is a map, called the \emph{labeling function}, such that $\lambda_{\A} (d_i^0x) = \lambda_{\A} (d_i^1x)$ for all $x \in (P_{\A})_2$ and $i \in \{1,2\}$ \cite{vanGlabbeek}. The vertices of an HDA are also called its \emph{states}. 
Originally, an HDA is also equipped with a set of final states, but since we will never need final states, we omit this part of the structure. HDAs form a category in which a morphism from an HDA \(\A\) to an HDA \(\B\) is a pair \((f,g)\) consisting of a morphism of precubical sets \(f\colon P_\A \to P_\B\) and a map \(g \colon \Sigma_\A \to \Sigma_\B\) such that \(f(I_{\A}) = I_{\B}\) and \(\lambda_{\B}(f(x)) = g(\lambda_{\A}(x))\) for all \(x \in (P_\A)_1\).

We turn the precubical set \(P\) defined in the previous section into an HDA \(\P\) by setting \(P_\P = P\), \(I_\P = (\{1\},\{1\})\), \(\Sigma_\P = \{1, \dots , N\}\), and \(\lambda_\P(\tau, \sigma) = a\)
for \((\tau, \sigma) \in P_1\) with \(\sigma \setminus \tau = \{a\}\). This is indeed an HDA since for \((\tau, \sigma)\in P_2\) with \(\sigma \setminus \tau = \{a < b\}\),
\begin{align*}
	\lambda_\P(d^0_1(\tau, \sigma)) = \lambda_\P(\tau, \tau \cup\{b\}) = b = \lambda_\P(\tau \cup \{a\}, \sigma) = \lambda_\P(d^1_1(\tau, \sigma))
\end{align*}
and 
\begin{align*}
	\lambda_\P(d^0_2(\tau, \sigma)) = \lambda_\P(\tau, \tau \cup\{a\}) = a = \lambda_\P(\tau \cup \{b\}, \sigma) = \lambda_\P(d^1_2(\tau, \sigma)).
\end{align*}

An HDA is said to be \emph{deterministic} if no two edges with the same source have the same label. An HDA is said to be \emph{codeterministic} if no two edges with the same target have the same label. We say that an HDA is \emph{bideterministic} if it is both deterministic and codeterministic.

\begin{prop} \label{Pbidet}
	The HDA \(\P\) is bideterministic.
\end{prop}

\begin{proof}
	Let \((\tau, \tau)\) be a vertex of \(\P\), and let \((\tau, \sigma)\) and \((\tau, \rho)\) be two edges with the same label starting in \((\tau, \tau)\). Suppose that \(\sigma \setminus \tau = \{a\}\) and that \(\rho \setminus \tau = \{b\}\). Then \(a = \lambda_\P(\tau, \sigma) = \lambda _\P(\tau, \rho) = b\). Hence \(\sigma = \tau \cup \{a\} = \tau \cup \{b\} = \rho\), and so the two edges are the same. Thus, \(\P\) is deterministic. A similar argument shows that \(\P\) is codeterministic.	
\end{proof}

\subsection*{HDA models}

An HDA is \emph{extensional} if no two edges with the same endpoints have the same label. If an HDA is deterministic or codeterministic, it is extensional. A \emph{transition system} is a 1-truncated extensional HDA, i.e., an extensional HDA concentrated in degrees \(\leq 1\).

Let \(\T\) be a transition system, and let \(R\) be a relation on the alphabet \(\Sigma_\T\).  The \emph{HDA model of \(\T\) with respect to \(R\)} is the by \cite[Thm. 4.2, Cor. 4.5]{transhda} up to isomorphism uniquely determined HDA ${\mathcal{Q}}$ which satisfies the following conditions:
\begin{description}
	\item[HM1] The \emph{1-skeleton} of \(\Q\), \(\Q_{\leq 1} = ((P_\Q)_{\leq 1}, I_\Q, \Sigma_\Q, \lambda_\Q)\), is \(\T\). 
	\item[HM2] For all $x \in (P_{\mathcal{Q}})_2$, $\lambda_{\mathcal{Q}}(d^0_2x) \,R\, \lambda_{\mathcal{Q}}(d^0_1x)$.		
	\item[HM3] For all $m\geq 2$ and $x,y \in (P_{\mathcal{Q}})_m$, if $d^k_rx = d^k_ry$ for all $r \in\{1,\dots ,m\}$ and $k \in \{0,1\}$, then $x = y$.		
	\item[HM4] ${\mathcal{Q}}$ is maximal with respect to the properties HM1--HM3, i.e., there is no HDA \(\A\) satisfying HM1--HM3 with \(P_\Q \subsetneqq P_\A\).	
\end{description}
Condition HM1 says that \(\Q\) is built on top of \(\T\) by filling in empty cubes. By condition HM2, an empty square may only be filled in if the labels of its edges are related. Condition HM3 ensures that no empty cube is filled in twice in the same way. By condition HM4, all admissible empty cubes are filled in. 

\begin{theor} \label{PHDAmod}
	\(\P\) is the HDA model of its 1-skeleton \(\P_{\leq 1}\) with respect to \(<\).
\end{theor}

\begin{proof}
	Since \(\P\) is deterministic, \(\P_{\leq 1}\) is indeed a transition system. 
	
	HM1 is obvious. 
	
	HM2: Let \((\tau, \sigma)\) be a 2-cube of \(\P\). Suppose that \(\sigma \setminus \tau = \{a < b\}\). We have 
	\[\lambda_\P(d^0_2(\tau, \sigma)) = \lambda_\P(\tau,\tau \cup \{a\}) = a < b = \lambda_\P(\tau,\tau \cup \{b\}) = \lambda_\P(d^0_1(\tau, \sigma)).\]
	
	HM3: Let \(m \geq 2\), and let \((\tau, \sigma), (\tau', \sigma')\in P_m\) such that  \(d^k_r(\tau, \sigma) = d^k_r(\tau',\sigma')\) for all \(r \in\{1,\dots ,m\}\) and \(k \in \{0,1\}\). Suppose that \(\sigma \setminus \tau = \{w_1 < \dots < w_m\}\) and that \(\sigma' \setminus \tau' = \{w_1' < \dots < w_m'\}\). Since \(d^0_1(\tau, \sigma) = d^0_1(\tau',\sigma')\), we have 
	\((\tau, \sigma \setminus \{w_1\}) = (\tau', \sigma' \setminus \{w_1'\})\) and therefore \(\tau = \tau'\). Since \(d^1_1(\tau, \sigma) = d^1_1(\tau',\sigma')\), we have 
	\((\tau \cup \{w_1\}, \sigma ) = (\tau' \cup \{w_1'\}, \sigma' )\) and therefore \(\sigma = \sigma'\). Thus, \((\tau, \sigma) = (\tau', \sigma')\). 
	
	HM4: Let \(\Q\) be an HDA satisfying HM1--HM3 with respect to \(\P_{\leq 1}\) and \(<\) such that \(P_\P = P\) is a precubical subset of \(P_\Q\). By HM1, \((P_\Q)_m = P_m\) for \(m \leq 1\). Let \(m \geq 2\), and suppose inductively that \((P_\Q)_{m-1} = P_{m-1}\). Let \(x \in (P_\Q)_m\). By the inductive hypothesis, \(d^k_ix \in P_{m-1}\) for all \(i\) and \(k\). Write \(d^k_ix = (\tau^k_i, \sigma^k_i)\) and \(\sigma^k_i \setminus \tau^k_i = \{w^k_{i,1} < \dots < w^k_{i,m-1}\}\). Consider \(i \in \{1, \dots, m\}\) and \(j \in \{1, \dots, m-1\}\). If \(m = 2\), we have 
	\begin{align*}
		w^0_{i,j} &= w^0_{i,1} = \lambda_\P(\tau^0_i,\tau^0_i\cup \{w^0_{i,1}\}) = \lambda_\P(\tau^0_i,\sigma^0_i) = 
		\lambda_\P(d^0_ix)\\ &= \lambda_\P(d^1_ix) = \lambda_\P(\tau^1_i, \sigma^1_i) = \lambda_\P(\tau^1_i, \tau^1_i \cup \{w^1_{i,1}\}) = w^1_{i,1} = w^1_{i,j}.
	\end{align*}
	If \(m \geq 3\), we have
	\begin{align*}
		\MoveEqLeft{
		\lambda_\P(d^0_1\cdots d^0_{j-1}d^0_{j+1} \cdots d^0_{m-1}d^0_ix)}\\ &= \lambda_\P(d^0_1\cdots d^0_{j-1}d^0_{j+1} \cdots d^0_{m-1}(\tau^0_i,\tau^0_i\cup \{w^0_{i,1}< \dots < w^0_{i,m-1}\}))\\
		&= \lambda_\P(d^0_1\cdots d^0_{j-1}d^0_{j+1} \cdots d^0_{m-2}(\tau^0_i,\tau^0_i\cup \{w^0_{i,1} < \dots < w^0_{i,m-2}\}))\\
		&= \cdots\\
		&= \lambda_\P(d^0_1\cdots d^0_{j-1}d^0_{j+1} (\tau^0_i,\tau^0_i\cup \{w^0_{i,1} < \dots < w^0_{i,j-1}< w^0_{i,j}< w^0_{i,j+1} \}))\\
		&= \lambda_\P(d^0_1\cdots d^0_{j-1} (\tau^0_i,\tau^0_i\cup \{w^0_{i,1} < \dots < w^0_{i,j-1}< w^0_{i,j} \}))\\
		&= \cdots\\
		&= \lambda_\P(\tau^0_i,\tau^0_i\cup \{w^0_{i,j} \})\\
		&= w^0_{i,j}		
	\end{align*}
	and 
	\begin{align*}
		\MoveEqLeft{
		\lambda_\P(d^1_1\cdots d^1_{j-1}d^1_{j+1} \cdots 	d^1_{m-1}d^1_ix)}\\ &= \lambda_\P(d^1_1\cdots d^1_{j-1}d^1_{j+1} \cdots d^1_{m-1}(\tau^1_i,\tau^1_i\cup \{w^1_{i,1}< \dots < w^1_{i,m-1}\}))\\
		&= \lambda_\P(\tau^1_i\cup \{w^1_{i,1}< \dots < w^1_{i,j-1}< w^1_{i,j+1}< \cdots  w^1_{i,m-1}\},\\ &\quad \quad \quad \tau^1_i\cup \{w^1_{i,1}< \dots < w^1_{i,m-1}\})\\
		&= w^1_{i,j}.
	\end{align*}
	Since parallel edges in an HDA have the same label (see, e.g., \cite[Lemma 4.6]{transhda}, it follows that \(w^0_{i,j} = w^1_{i,j}\) in the case \(m \geq 3\) as well.

	Let \(1 \leq i < j \leq m\). Since 
	\begin{align*}
		d^0_id^0_jx &= d^0_i(\tau^0_j, \sigma^0_j) = (\tau^0_j, \sigma^0_j\setminus \{w^0_{j,i}\})
	\end{align*}
	and 
	\begin{align*}
		d^0_{j-1}d^0_ix &= d^0_{j-1}(\tau^0_i, \sigma^0_i) = (\tau^0_i, \sigma^0_i\setminus \{w^0_{i,j-1}\}),
	\end{align*}
	we have \(\tau^0_i = \tau^0_j\). Set \(\tau = \tau^0_i = \tau^0_j\). Since 
	\begin{align*}
		d^0_id^1_jx &= d^0_i(\tau^1_j, \sigma^1_j) = (\tau^1_j, \sigma^1_j\setminus \{w^1_{j,i}\})
	\end{align*}
	and 
	\begin{align*}
		d^1_{j-1}d^0_ix &= d^1_{j-1}(\tau^0_i, \sigma^0_i) = (\tau^0_i \cup \{w^0_{i,j-1}\}, \sigma^0_i),
	\end{align*}
	we have \(\tau^1_j = \tau \cup \{w^0_{i,j-1}\}\). Since 
	\begin{align*}
		d^1_id^0_jx &= d^1_i(\tau^0_j, \sigma^0_j) = (\tau^0_j \cup \{w^0_{j,i}\}, \sigma^0_j)
	\end{align*}
	and 
	\begin{align*}
		d^0_{j-1}d^1_ix &= d^0_{j-1}(\tau^1_i, \sigma^1_i) = (\tau^1_i, \sigma^1_i \setminus \{w^1_{i,j-1}\}),
	\end{align*}
	we have \(\tau^1_i = \tau \cup \{w^0_{j,i}\}\).
	
	Since \(\tau^1_j = \tau \cup \{w^0_{i,j-1}\}\) for all \(1\leq  i < j \leq m\), we have 
	\[w^0_{1,j-1} = w^0_{2,j-1} = \dots  = w^0_{j-1,j-1}\]
	for all \(1 < j \leq m\). Since \(\tau^1_i = \tau \cup \{w^0_{j,i}\}\) for all \(1\leq  i < j \leq m\), we have 
	\[w^0_{i+1,i} = w^0_{i+2,i} = \dots  = w^0_{m,i}\]
	for all \(1\leq i < m\). Since  \(\tau \cup \{w^0_{i+1,i}\} = \tau^1_i =  \tau \cup \{w^0_{i-1,i-1}\}\) for all \(1 < i < m\), we have 
	\begin{align*}
		w^0_{1,i-1} = w^0_{2,i-1} = \dots  = w^0_{i-1,i-1} = w^0_{i+1,i} = w^0_{i+2,i} = \dots  = w^0_{m,i}
	\end{align*} 
	for all \(1 < i < m\).
 
Set 
\[w_i = \begin{cases}
	w^0_{i+1,i}, & 1 \leq i < m,\\
	w^0_{1, m-1}, & i = m.
\end{cases}
\]
Then \(w_1 < \dots < w_m\). Indeed, if \(m = 2\), since \(\Q\) satisfies HM1 and HM2,
\begin{align*}
	w_1 &= w^0_{2,1} = \lambda_\P(\tau^0_2, \tau^0_2 \cup \{w^0_{2,1}\}) = \lambda_\P(d^0_2x) =  \lambda_\Q(d^0_2x)\\
	&< \lambda_\Q(d^0_1x) = \lambda_\P(d^0_1x) = \lambda_\P(\tau^0_1, \tau^0_1 \cup \{w^0_{1,1}\}) = w^0_{1,1} = w_2.
\end{align*}
If \(m\geq 3\), we have 
\[w_{m-1} = w^0_{m,m-1} = w^0_{m-2,m-2} < w^0_{m-2,m-1} = w^0_{1,m-1} = w_m\]
and
\[w_i = w^0_{i+1, i} = w^0_{i+2,i} < w^0_{i+2,i+1} = w_{i+1}\]
for \(1\leq i  < m-1\).

We have 
\begin{align*}
	\MoveEqLeft{d^0_m(\tau, \tau \cup \{w_1 < \dots  < w_m\})}\\ &= (\tau, \tau \cup \{w_1 <  \dots < w_{m-1}\})\\
	&= (\tau, \tau \cup \{w^0_{2,1} < \dots <  w^0_{m, m-1}\})\\
	&= (\tau, \tau \cup \{w^0_{m,1} < \dots <  w^0_{m, m-1}\})\\
	&= (\tau^0_m, \sigma^0_m)\\
	&= d^0_mx
\end{align*}
and 
\begin{align*}
	\MoveEqLeft{d^1_m(\tau, \tau \cup \{w_1 < \dots  < w_m\})}\\ &= (\tau \cup \{w_m\}, \tau \cup \{w_1 <  \dots < w_{m}\})\\
	&= (\tau \cup \{w_m\}, \tau \cup \{w_m\} \cup \{w_1 <  \dots < w_{m-1}\})\\
	&= (\tau \cup \{w^0_{1,m-1}\}, \tau \cup \{w^0_{1,m-1}\} \cup \{w^0_{2,1} <  \dots < w^0_{m,m-1}\})\\
	&= (\tau^1_m, \tau^1_m \cup \{w^0_{m,1} <  \dots < w^0_{m,m-1}\})\\
	&= (\tau^1_m, \tau^1_m \cup \{w^1_{m,1} <  \dots < w^1_{m,m-1}\})\\
	&= (\tau^1_m, \sigma^1_m)\\
	&= d^1_mx.
\end{align*}
For \(1 \leq i < m\), we have
\begin{align*}
	\MoveEqLeft{d^0_i(\tau, \tau \cup \{w_1 < \dots < w_m\})}\\ &= (\tau, \tau \cup \{w_1 < \dots w_{i-1} < w_{i+1} < \dots < w_{m-1} <  w_m\})\\
	&= (\tau, \tau \cup \{w^0_{2,1} < \dots < w^0_{i,i-1} < w^0_{i+2,i+1} < \dots <  w^0_{m, m-1} < w^0_{1,m-1}\})\\
	&= (\tau, \tau \cup \{w^0_{i,1} < \dots < w^0_{i,i-1} < w^0_{i,i} < \dots <  w^0_{i, m-2} < w^0_{i,m-1}\})\\
	&= (\tau^0_i, \sigma^0_i)\\
	&= d^0_ix
\end{align*}
and 
\begin{align*}
	\MoveEqLeft{d^1_i(\tau, \tau \cup \{w_1 < \dots < w_m\})}\\ &= (\tau \cup \{w_i\}, \tau \cup \{w_1 < \dots <  w_m\})\\
	&= (\tau \cup \{w^0_{i+1,i}\}, \tau \cup \{w^0_{2,1} <  \dots <  w^0_{m, m-1} < w^0_{1,m-1}\})\\
	&= (\tau \cup \{w^0_{i+1,i}\}, \tau \cup \{w^0_{i+1,i}\} \\
	&\quad \quad \quad  \cup \{w^0_{2,1} < \dots < w^0_{i,i-1} < w^0_{i+2,i+1} < \dots <  w^0_{m, m-1} < w^0_{1,m-1}\})\\
	&= (\tau \cup \{w^0_{i+1,i}\}, \tau \cup \{w^0_{i+1,i}\} \\
	&\quad \quad \quad 	\cup \{w^0_{i,1} < \dots < w^0_{i,i-1} < w^0_{i,i} < \dots <  w^0_{i, m-2} < w^0_{i,m-1}\})\\
	&= (\tau^1_i, \tau^1_i\cup \{w^1_{i,1} < \dots < w^1_{i,i-1} < w^1_{i,i} < \dots <  w^1_{i, m-2} < w^1_{i,m-1}\})\\
	&= (\tau^1_i, \sigma^1_i)\\
	&= d^1_ix.
\end{align*}
Since \(\Q\) satisfies HM3, it follows that \(x = (\tau, \tau \cup \{w_1 < \dots < w_m\}) \in P_m\). Thus, \((P_\Q)_m = P_m\).
\end{proof}

\section{Accessibility} \label{SecAcc}

A state \(v\) in HDA is said to be \emph{reachable} if there exists a path, i.e., a sequence of consecutive edges, from the initial state to \(v\). An HDA in which all states are reachable is called \emph{accessible}. Unreachable states are of very limited interest for the analysis of concurrent systems, since the executions of a system only pass through reachable states. Therefore, it makes sense to model only the accessible part of the  state space of a system. Another important reason for doing so is the state explosion problem: the state space of a concurrent system can easily become very large, and including unreachable states in the model would dramatically aggravate this problem. Unfortunately, the HDA \(\P\) defined in the previous section is highly inaccessible. In this section, we show that it is possible to modify \(\P\) to obtain an accessible HDA of the same homotopy type. More precisely, we prove the following theorem:

\begin{theor} \label{thmacc}
	Let \(\A\) be a bideterministic HDA which is the HDA model of its 1-skeleton with respect to a strict total order on \(\Sigma_\A\). Suppose that \(\A\) is connected, i.e., \(|P_\A|\) is path-connected, and that \(\A\) has only a finite number of unreachable states (e.g., \(\A\) is finite). Then there exists an accessible and bideterministic HDA \(\B\) which is the HDA model of its 1\(\mbox{-}\)skeleton with respect to a strict total order on \(\Sigma_\B\) and satisfies \({|P_\B| \simeq |P_\A|}\).
\end{theor}

For the proof, we may suppose that \(\A\) is not accessible. Clearly, it is enough to show that there exists a bideterministic HDA \(\B\) with less unreachable states than \(\A\) that is the HDA model of its 1-skeleton with respect to a strict total order on \(\Sigma_\B\) and satisfies \({|P_\B| \simeq |P_\A|}\). We show first that \(\A\) admits an edge from an unreachable to a reachable state. Suppose that there is no such edge. Let \(v\) be an unreachable state. Since \(\A\) is connected, there is a sequence of vertices \(I_\A = v_1, v_2, \dots, v_k = v\) such that for each \(1 \leq i < k\) there exists an edge between \(v_i\) and \(v_{i+1}\). Inductively, all \(v_i\) are reachable, which is impossible.

Let \(e\) be an edge of \(\A\) from an unreachable state \(v\) to a reachable state \(w\). If \(w = I_\A\), we define \(\B\) to be the same as \(\A\) but with \(I_\B = v\). Suppose that \(w \not= I_\A\). Let \(\lambda_\A(e) = a\), and let \(\omega = (x_1,\dots, x_k)\) be a path from \(I_\A\) to \(w\) with no repeated vertices, e.g., a shortest possible path. We view \(\omega\) as a morphism of precubical sets \(\lbrbrak 0, k \rbrbrak \to P_\A\), where the \emph{precubical interval} \(\lbrbrak p, q\rbrbrak\) (\(p,q \in \Z\), \(p\leq q\)) is the precubical set defined by \(\lbrbrak p,q \rbrbrak_0 = \{p,\dots , q\}\), \(\lbrbrak p,q \rbrbrak_1 =  \{{[p,p+1]}, \dots , {[q- 1,q]}\}\), \(d_1^0[j-1,j] = j-1\), \(d_1^1[j-1,j] = j\), and \(\lbrbrak p,q \rbrbrak_{n} = \emptyset\) for \(n > 1\).

\subsection*{The HDA \(\C\)}

We first extend \(\A\) to an HDA \(\C\) such that \({|P_\C| \simeq |P_\A|}\). We define the precubical set \(P_\C\) by the pushout diagram
\[
\begin{tikzcd}
	\lbrbrak 0, k\rbrbrak \otimes \{2\} \ar[r, "\cong"] \ar[d, hook] & \lbrbrak 0, k\rbrbrak  \ar[r, "\omega"]  & P_\A \ar[d, hook]\\
	\lbrbrak -1,0\rbrbrak \otimes \{1\} \cup \lbrbrak 0, k\rbrbrak \otimes \lbrbrak 1,2 \rbrbrak \ar[rr, "\xi"']
	& & P_\C.
\end{tikzcd}
\]
Since the geometric realizations of the precubical sets on the left are contractible and, as is well known, the geometric realization functor preserves colimits, the inclusion \(|P_\A| \hookrightarrow |P_\C|\) is a homotopy equivalence. Let \(\Sigma_\C = \Sigma_\A \cup \{c\}\) for some element \(c \notin \Sigma_\A\). We extend the labeling function of \(\A\) to \(\C\) by setting \(\lambda_\C(\xi(i,[1,2])) = c\) \((i \in\{0, \dots, k\})\),  \(\lambda_\C(\xi([i-1,i],1)) = \lambda_\A(x_i)\) \((i \in\{1, \dots, k\})\), and \(\lambda_\C(\xi([-1,0],1)) = a\). The initial state of \(\C\) is \({I_\C = \xi(-1,1)}\).
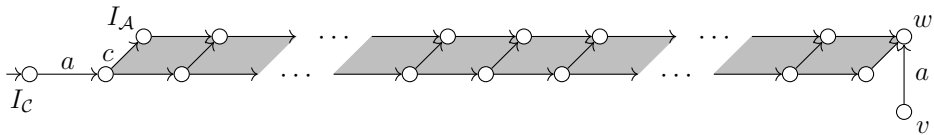
\begin{figure}[h!]
	\center
	\begin{tikzpicture}[initial text={},on grid] 
		
		\path[draw=lightgray, fill=lightgray] (-0.5,-0.5)--(0,0)--(2,0)--(1.5,-0.5)--cycle; 
		
		\path[draw=lightgray, fill=lightgray] (2.5,-0.5)--(3,0)--(7,0)--(6.5,-0.5)--cycle; 
		
		\path[draw=lightgray, fill=lightgray] (7.5,-0.5)--(8,0)--(10,0)--(9.5,-0.5)--cycle; 
		
		\node[state,minimum size=0pt,inner sep =2pt,fill=white] (q_0) at (0,0)  {}; 
		
		\node[state,minimum size=0pt,inner sep =2pt,fill=white] (q_1) [right=of q_0,xshift=0cm] {};
		
		\node[state,minimum size=0pt,inner sep =2pt,fill=white] (q_2)  at (4,0)  {};
		
		\node[state,minimum size=0pt,inner sep =2pt,fill=white] (q_3) [right=of q_2,xshift=0cm] {};
		
		\node[state,minimum size=0pt,inner sep =2pt,fill=white] (q_4) [right=of q_3,xshift=0cm] {};
		
		\node[state,minimum size=0pt,inner sep =2pt,fill=white] (q_5) at (9,0)  {}; 
		
		\node[state,minimum size=0pt,inner sep =2pt,fill=white] (q_6) [right=of q_5,xshift=0cm] {};

		\node[state,minimum size=0pt,inner sep =2pt,fill=white] (v) [below=of q_6,xshift=0cm] {};

		\node[state,minimum size=0pt,inner sep =2pt,fill=white, initial,initial where=left,initial distance=0.2cm] (p) at (-1.5,-0.5)  {}; 
		
		\node[state,minimum size=0pt,inner sep =2pt,fill=white] (p_0) at (-0.5,-0.5)  {}; 
		
		\node[state,minimum size=0pt,inner sep =2pt,fill=white] (p_1) [right=of p_0,xshift=0cm] {};
		
		\node[state,minimum size=0pt,inner sep =2pt,fill=white] (p_2)  at (3.5,-0.5)  {};
		
		\node[state,minimum size=0pt,inner sep =2pt,fill=white] (p_3) [right=of p_2,xshift=0cm] {};
		
		\node[state,minimum size=0pt,inner sep =2pt,fill=white] (p_4) [right=of p_3,xshift=0cm] {};
		
		\node[state,minimum size=0pt,inner sep =2pt,fill=white] (p_5) at (8.5,-0.5)  {}; 
		
		\node[state,minimum size=0pt,inner sep =2pt,fill=white] (p_6) [right=of p_5,xshift=0cm] {};

		\path[->] 
		(q_0) edge (q_1)
		(q_1) edge (2,0)
		(3,0) edge (q_2)
		(q_2) edge (q_3)
		(q_3) edge (q_4)
		(q_4) edge (7,0)
		(8,0) edge (q_5)
		(q_5) edge (q_6)
		(v)   edge[right] node{\scalebox{0.85}{\(a\)}} (q_6)
		
		(p)   edge (p_0)
		(p_0) edge (p_1)
		(p_0) edge[left] node{\scalebox{0.85}{\(c\)}} (q_0)
		(p_1) edge (1.5,-0.5)
		(p_1) edge (q_1)
		(2.5,-0.5) edge (p_2)
		(p_2) edge (p_3)
		(p_2) edge (q_2)
		(p_3) edge (p_4)
		(p_3) edge (q_3)
		(p_4) edge (6.5,-0.5)
		(p_4) edge (q_4)
		(7.5,-0.5) edge (p_5)
		(p_5) edge (p_6)
		(p_5) edge (q_5)
		(p_6) edge (q_6)
		;
		
		\node at (2.5,-0.025) {\scalebox{0.85}{\(\cdots\)}};
		\node at (2,-0.525) {\scalebox{0.85}{\(\cdots\)}};
		\node at (7.5,-0.025) {\scalebox{0.85}{\(\cdots\)}};
		\node at (7,-0.525) {\scalebox{0.85}{\(\cdots\)}};
		\node at (-1,-0.35) {\scalebox{0.85}{\(a\)}};
		\node at (-0.3,0.25) {\scalebox{0.85}{\(I_\A\)}};
		\node at (-1.6,-0.85) {\scalebox{0.85}{\(I_\C\)}};
		\node at (10.25,-1.2) {\scalebox{0.85}{\(v\)}};
		\node at (10.25, 0.2) {\scalebox{0.85}{\(w\)}};

	\end{tikzpicture}
	
\end{figure}

\begin{lem} \label{Cbidet}
	\(\C\) is  bideterministic and has the same unreachable states as \(\A\).
\end{lem}

\begin{proof}
	Since all edges of \(\C\) that are not edges of \(\A\) start in vertices of \(\C\) that are not vertices of \(\A\) and in no such vertex start two edges with the same label, \(\C\) is deterministic. Since \(\omega\) has no repeated vertices, no two edges of \(\C\) that are not edges of \(\A\) end in the same vertex. Since any such edge that ends in a vertex of \(\A\) has label \(c\), it follows that \(\C\) is codeterministic. 
	
	Since \(I_\A\) is reachable in \(\C\), all states of \(\A\) that are reachable in \(\A\) are also reachable in \(\C\). On the other hand, since any path in \(\C\) from \(I_\C\) to a state of \(\A\) intersects \(\omega\), all states of \(\A\) that are reachable in \(\C\) are also reachable in \(\A\). Since all states in \(\xi(\lbrbrak -1,0\rbrbrak \otimes \{1\} \cup \lbrbrak 0, k\rbrbrak \otimes \lbrbrak 1,2 \rbrbrak)\) are reachable in \(\C\), it follows that \(\C\) has the same unreachable states as \(\A\).
\end{proof}

Let \(<\) be the strict total order on \(\Sigma_\A\) with respect to which \(\A\) is the HDA model of its 1-skeleton. We extend \(< \) to a strict total order on \(\Sigma_\C\) by setting \(b< c\) for all \(b\in \Sigma_\A\).  

\begin{lem} \label{CHM}
\(\C\) is the HDA model of \(\C_{\leq 1}\) with respect to \(<\).	
\end{lem}

\begin{proof}
	Condition HM1 is trivially satisfied. Since \(\A\) satisfies HM2 and, for all \(i \in \{1,\dots, k\}\), 
	\begin{align*}
		\lambda_\C(d^0_2\xi([i-1,i],[1,2])) &=  \lambda_\C(\xi(d^0_2([i-1,i],[1,2])))\\
		&= \lambda_\C(\xi([i-1,i],1))\\
		&= \lambda_\A(x_i)\\
		&<  c\\
		&= \lambda_\C(\xi(i-1,[1,2]))\\
		&= \lambda_\C(\xi(d^0_1([i-1,i],[1,2])))\\
		&= \lambda_\C(d^0_1\xi([i-1,i],[1,2])),		
	\end{align*}
	\(\C\) satisfies HM2. Let $x,y \in (P_\C)_m$ \((m\geq 2)\) such that $d^k_rx = d^k_ry$ for all $r \in\{1,\dots ,m\}$ and $k \in \{0,1\}$. Since \(\A\) satisfies HM3, \(x = y\) if \(x, y \in (P_\A)_m\). If \(x \notin (P_\A)_m\), then \(m = 2\) and \(x = \xi([i-1,i],[1,2])\) for some \(i \in \{1, \dots, k\}\). Since \(x\) is the only 2-cube of \(\C\) having \(\xi([i-1,i],1)\) in its boundary, \(y = x\). Hence \(\C\) satisfies HM3. Let \(\Q\) be an HDA with \(P_\C \subseteq P_\Q\) that satisfies HM1--HM3. Since \(\A\) is the HDA model of \(\A_{\leq 1}\), all \(m\)-cubes of \(\Q\) \((m\geq 2)\)  with faces in \(\A\) belong to \(\A\). Let \(x \in (P_\Q)_2\) such that at least one edge of \(x\) does not belong to \(\A\). Since every edge that starts in a vertex of \(\A\) belongs to \(\A\), \(d^0_1d^0_2x\) is not a vertex of \(\A\). Hence \(d^0_1d^0_2x = \xi(i,1)\) for some \(i \in\{-1,\dots, k\}\). Since \(\lambda_\Q(d^0_2x)<  \lambda_\Q(d^0_1x)\), we have \(d^0_1x \not= d^0_2x\) and therefore \(0 \leq i < k\). Since \(\lambda_\Q(\xi([i,i+1],1)) < \lambda_\Q (\xi(i,[1,2])) = c\), we have \(d^0_1x = \xi(i,[1,2])\) and \(d^0_2x = \xi([i,i+1],1)\). Since \(\lambda_\Q(d^1_1x) = \lambda_\Q(d^0_1x) = c\), we have \(d^1_1x = \xi(i+1,[1,2])\). Since, by Lemma \ref{Cbidet}, \(\Q_{\leq 1} = \C_{\leq 1}\) is deterministic, \(d^1_2x = x_{i+1}\) because \(d^1_2x\) starts in \(\xi(i,2) = d^0_1x_{i+1}\) and \(\lambda_\Q(d^1_2x) = \lambda_\Q(d^0_2x) = \lambda_\Q(\xi([i,i+1],1)) = \lambda_\Q(x_{i+1})\). Since \(\Q\) satisfies HM3, \(x = \xi([i,i+1],[1,2]) \in (P_\C)_2\). Suppose that there exist an integer \(m\geq 3\) and an element \(y \in (P_\Q)_m\) such that at least one face of \(y\) does not belong to \(\A\). Then \(d^0_1\cdots d^0_my\) is not a vertex of \(\A\). Hence \(d^0_1\cdots d^0_my = \xi(i,1)\) for some \(i\). Since, by \cite[Prop. 4.7]{transhda}, for all \(1 \leq i < j \leq m\), 
	\[\lambda_\Q(d^0_1\cdots d^0_{i-1}d^0_{i+1}\cdots d^0_mx) < \lambda_\Q(d^0_1\cdots d^0_{j-1}d^0_{j+1}\cdots d^0_mx),\]
	\(m\) different edges start in \(\xi(i,1)\). This is not the case. Thus, \(\Q = \C\) and \(\C\) satisfies HM4.
\end{proof}

\subsection*{The HDA \(\D\)}

\begin{sloppypar}
We now extend \(\C\) to an HDA \(\D\) that still satisfies \({|P_\D| \simeq |P_\A|}\) but in which \(v\) is reachable. Let \(i\) be the largest index in \(\{0, \dots, k\}\) such that \({\lambda_\C(\xi([i-1,i],1)) = a}\). Since \(\A\) is codeterministic, \(\lambda_\C(\xi([k-1,k],1)) = \lambda_\A(x_k) \not= a\). Hence \(i < k\). We define the precubical set \(P_\D\) by the pushout diagram
\[\begin{tikzcd}
	\lbrbrak i, k+1\rbrbrak \otimes \{1\} \cup \{i,k+1\}\otimes \lbrbrak 0,1\rbrbrak \ar[r, "\nu"] \ar[d, hook]  & P_\C \ar[d, hook]\\
	\lbrbrak i, k+1\rbrbrak \otimes \lbrbrak 0,1 \rbrbrak \ar[r,"\chi"']
	&  P_\D
\end{tikzcd}
\]
where \(\nu\) is the unique morphism of precubical sets such that \(\nu([j-1,j],1) = \xi([j-1,j],1)\) \((i < j \leq k)\), \(\nu([k,k+1],1) = \xi(k,[1,2])\), \(\nu(i,[0,1]) = {\xi([i-1,i],1)}\), and \(\nu(k+1,[0,1]) = e\). Note that \(\nu\) is injective. Since the geometric realizations of the precubical sets on the left are contractible, the inclusion \(|P_\C| \hookrightarrow |P_\D|\) is a homotopy equivalence. Hence \(|P_\D| \simeq |P_\A|\). We set \(I_\D = I_\C\) and \(\Sigma_\D = \Sigma_\C\) and extend the labeling function of \(\C\) to \(\D\) by setting \(\lambda_\D(\chi(j,[0,1])) = a\), \({\lambda_\D(\chi([j-1,j],0))} = \lambda_\A(x_j)\) \((i < j \leq k)\),  and \(\lambda_\D(\chi([k,k+1],0)) = c\). 
\begin{figure}[h!]
	\center
	\begin{tikzpicture}[initial text={},on grid] 
		
		\path[draw=lightgray, fill=lightgray] (-0.5,-0.5)--(0,0)--(2,0)--(1.5,-0.5)--cycle; 
		
		\path[draw=lightgray, fill=lightgray] (2.5,-0.5)--(3,0)--(7,0)--(6.5,-0.5)--cycle; 
		
		\path[draw=lightgray, fill=lightgray] (7.5,-0.5)--(8,0)--(10,0)--(9.5,-0.5)--cycle;
		
		\path[draw=lightgray, fill=lightgray] (3.5,-0.5)--(6.5,-0.5)--(6.5,-1.5)--(5.5,-1.5)--cycle;
		
		\path[draw=lightgray, fill=lightgray] (7.5,-0.5)--(9.5,-0.5)--(9.5,-1.5)--(7.5,-1.5)--cycle;
		
		\path[draw=lightgray, fill=lightgray] (9.5,-0.5)--(10,0)--(10,-1)--(9.5,-1.5)--cycle;

		\node[state,minimum size=0pt,inner sep =2pt,fill=white] (q_0) at (0,0)  {}; 
		
		\node[state,minimum size=0pt,inner sep =2pt,fill=white] (q_1) [right=of q_0,xshift=0cm] {};
		
		\node[state,minimum size=0pt,inner sep =2pt,fill=white] (q_2)  at (4,0)  {};
		
		\node[state,minimum size=0pt,inner sep =2pt,fill=white] (q_3) [right=of q_2,xshift=0cm] {};
		
		\node[state,minimum size=0pt,inner sep =2pt,fill=white] (q_4) [right=of q_3,xshift=0cm] {};
		
		\node[state,minimum size=0pt,inner sep =2pt,fill=white] (q_5) at (9,0)  {}; 
		
		\node[state,minimum size=0pt,inner sep =2pt,fill=white] (q_6) [right=of q_5,xshift=0cm] {};

		\node[state,minimum size=0pt,inner sep =2pt,fill=white] (v) [below=of q_6,xshift=0cm] {};

		\node[state,minimum size=0pt,inner sep =2pt,fill=white, initial,initial where=left,initial distance=0.2cm] (p) at (-1.5,-0.5)  {}; 
		
		\node[state,minimum size=0pt,inner sep =2pt,fill=white] (p_0) at (-0.5,-0.5)  {}; 
		
		\node[state,minimum size=0pt,inner sep =2pt,fill=white] (p_1) [right=of p_0,xshift=0cm] {};
		
		\node[state,minimum size=0pt,inner sep =2pt,fill=white] (p_2)  at (3.5,-0.5)  {};
		
		\node[state,minimum size=0pt,inner sep =2pt,fill=white] (p_3) [right=of p_2,xshift=0cm] {};
		
		\node[state,minimum size=0pt,inner sep =2pt,fill=white] (p_4) [right=of p_3,xshift=0cm] {};
		
		\node[state,minimum size=0pt,inner sep =2pt,fill=white] (p_5) at (8.5,-0.5)  {}; 
		
		\node[state,minimum size=0pt,inner sep =2pt,fill=white] (p_6) [right=of p_5,xshift=0cm] {};

		\node[state,minimum size=0pt,inner sep =2pt,fill=white] (r_4) [below=of p_4,xshift=0cm]  {};
		
		\node[state,minimum size=0pt,inner sep =2pt,fill=white] (r_5) [below=of p_5,xshift=0cm]  {};
		
		\node[state,minimum size=0pt,inner sep =2pt,fill=white] (r_6) [below=of p_6,xshift=0cm]  {};
		
		\path[->] 
		(q_0) edge (q_1)
		(q_1) edge (2,0)
		(3,0) edge (q_2)
		(q_2) edge (q_3)
		(q_3) edge (q_4)
		(q_4) edge (7,0)
		(8,0) edge (q_5)
		(q_5) edge (q_6)
		(v)   edge[right] node{\scalebox{0.85}{\(a\)}} (q_6)
		
		(p)   edge (p_0)
		(p_0) edge (p_1)
		(p_0) edge[left] node{\scalebox{0.85}{\(c\)}} (q_0)
		(p_1) edge (1.5,-0.5)
		(p_1) edge (q_1)
		(2.5,-0.5) edge (p_2)
		(p_2) edge (p_3)
		(p_2) edge (q_2)
		(p_2) edge (r_4)
		(p_3) edge (p_4)
		(p_3) edge (q_3)
		(p_4) edge (6.5,-0.5)
		(p_4) edge (q_4)
		(7.5,-0.5) edge (p_5)
		(p_5) edge (p_6)
		(p_5) edge (q_5)
		(p_6) edge (q_6)
		
		(r_4) edge (6.5,-1.5)
		(r_4) edge (p_4)
		(7.5,-1.5) edge (r_5)
		(r_5) edge (r_6)
		(r_5) edge (p_5)
		(r_6) edge (p_6)
		(r_6) edge (v)
		;
		
		\node at (2.5,-0.025) {\scalebox{0.85}{\(\cdots\)}};
		\node at (2,-0.525) {\scalebox{0.85}{\(\cdots\)}};
		\node at (7.5,-0.025) {\scalebox{0.85}{\(\cdots\)}};
		\node at (7,-0.525) {\scalebox{0.85}{\(\cdots\)}};
		\node at (7,-1.525) {\scalebox{0.85}{\(\cdots\)}};
		\node at (-1,-0.35) {\scalebox{0.85}{\(a\)}};
		\node at (4.1,-0.35) {\scalebox{0.85}{\(a\)}};
		\node at (-0.3,0.25) {\scalebox{0.85}{\(I_\A\)}};
		\node at (-1.6,-0.85) {\scalebox{0.85}{\(I_\D = I_\C\)}};
		\node at (10.25,-1.2) {\scalebox{0.85}{\(v\)}};
		\node at (10.25, 0.2) {\scalebox{0.85}{\(w\)}};

	\end{tikzpicture}
	
\end{figure}
\end{sloppypar}

\begin{lem} \label{D}
	\(\D\) is bideterministic and has less unreachable states than \(\A\).
\end{lem}

\begin{proof}
	In each vertex of \(\D\) that is not a vertex of \(\C\) start exactly two edges, one with label \(a\) and the other with a different label. The only vertex of \(\C\) in which starts an edge of \(\D\) that is not an edge of \(\C\) is \(\nu(i,0) = \xi(i-1,1)\). By definition of \(i\), the label of this edge is different from \(a\). Since \(i < k\), this label is also different from \(c\). Hence the edges starting in \(\nu(i,0) = \xi(i-1,1)\) have different labels. Since, by Lemma \ref{Cbidet}, \(\C\) is deterministic, it follows that \(\D\) is deterministic. Since no two edges of \(\D\) that are not edges of \(\C\) end in the same vertex, no edge in \(\C\) with label \(c\) ends in \(v\), and the edges \({\xi ([i,i+1],1), \dots, \xi ([k-1,k],1)}\) have labels different from \(a\), \(\D\) is codeterministic. 
	
	All states of \(\C\) that are reachable in \(\C\) are also reachable in \(\D\). Since all states in \(\chi(\lbrbrak i, k+1\rbrbrak \otimes \lbrbrak 0,1 \rbrbrak)\) are reachable in \(\D\) and, in particular, \(v = \chi(k+1,0)\) is reachable in \(\D\), the number of unreachable states of \(\D\) is less than the number of unreachable states of \(\C\) and hence, by Lemma \ref{Cbidet}, of \(\A\).
\end{proof}

\subsection*{The HDA \(\B\)}
\begin{sloppypar}
Unfortunately, we cannot guarantee that \(\D\) is the HDA model of its \(\mbox{1-}\)skeleton, because the labels of the edges of the squares added to \(\C\) might be related in the wrong way. In the final HDA \(\B\), we solve this problem. We set \((P_\B)_m = (P_\D)_m\) for all \(m\) and define the face maps of \(P_\B\) by
\[\partial ^k_i x = \left \{ \begin{array}{ll}
	d^k_{3-i}x, & x \in (P_\D)_2 \setminus (P_\C)_2, \, \lambda_\D(d^0_1x) < \lambda_\D(d^0_2x),\\
	d^k_ix, & \mbox{else}.
\end{array}\right.\]
Then \(P_\B\) is a precubical set and \(|P_\B| \approx |P_\D|\). Hence \(|P_\B| \simeq |P_\A|\). We set \(I_\B = I_\D = I_\C\), \(\Sigma_\B = \Sigma_\D = \Sigma_\C\), and \(\lambda_\B = \lambda_\D\). Then \(\B\) is an HDA with \(\B_{\leq 1} = \D_{\leq 1}\). By Lemma \ref{D}, \(\B\) is bideterministic and has less unreachable states than \(\A\). To finish the proof of Theorem \ref{thmacc}, it remains to show that \(\B\) is the HDA model of its 1-skeleton. This is done in Proposition \ref{BHM} below. 
\end{sloppypar}
\begin{lem} \label{facelem}
	Let \(\Q\) be an HDA which satisfies HM2 with respect to \(\Q_{\leq 1}\) and a strict total order on \(\Sigma_\Q\), and let \(x\) be an element of \(P_\Q\) of degree \(m\geq 3\). Then \(d^k_px \not= d^l_qx\) for all \(1 \leq p < q \leq m\) and \(k, l \in \{0,1\}\).
\end{lem}

\begin{proof}
    Suppose that \(d^k_px = d^l_qx\) for some \(1 \leq p < q \leq m\) and \(k, l \in \{0,1\}\). Then 
	\begin{align*}
		\MoveEqLeft{d^0_1 \cdots d^0_{p-1}d^k_pd^0_{p+2}\cdots d^0_{m}x = d^0_1 \cdots d^0_{p-1}d^0_{p+1}\cdots d^0_{m-1}d^k_px}\\ 
		&= d^0_1 \cdots d^0_{p-1}d^0_{p+1}\cdots d^0_{m-1}d^l_qx
		= d^0_1 \cdots d^0_{p-1}d^0_{p+1}\cdots d^0_{q-1}d^l_qd^0_{q+1}\cdots d^0_{m}x.
	\end{align*}
	By the arguments given in \cite[Lemma 4.6, Prop 4.7]{transhda}, it follows that 
	\begin{align*}
		\lambda_\Q(d^0_1 \cdots d^0_pd^0_{p+2}\cdots d^0_{m}x) &= \lambda_\Q(d^0_1 \cdots d^0_{p-1}d^0_{p+1}\cdots d^0_{m}x)\\
		&< \lambda_\Q(d^0_1 \cdots d^0_pd^0_{p+2}\cdots d^0_{m}x),
	\end{align*}
	which is impossible.
\end{proof}

\begin{prop} \label{BHM}
	\(\B\) is the HDA model of \(\B_{\leq 1}\) with respect to \(<\). 
\end{prop}

\begin{proof}
	By construction, \(\B\) satisfies HM1 and HM2. HM3 can be shown in a similar way as for \(\C\), see Lemma \ref{CHM}. Let \(\Q\) be an HDA that contains \(\B\) and satisfies HM1--HM3 with respect to \(\B_{\leq 1}\) and \(<\). Since, by Lemma \ref{CHM}, \(\C\) is the HDA model of \(\C_{\leq 1}\), all \(m\)-cubes of \(\Q\) \((m\geq 2)\)  with faces in \(\C\) belong to \(\C\). Let \(x \in (P_\Q)_2\) such that at least one edge of \(x\) does not belong to \(\C\). Since every edge with endpoints in \(\C\) belongs to \(\C\), \(x\) has a vertex that does not belong to \(\C\). Since \(\lambda_\Q(\partial^1_2x)< \lambda_\Q(\partial^1_1x)\), we have \(\partial^1_1x \not= \partial^1_2x\). Since in no vertex that does not belong to \(\C\) ends more than one edge, it follows that \(\partial^1_1\partial^1_2x \in (P_\C)_0\). This implies that if \(\partial^0_1\partial^0_2x \in (P_\C)_0\), then \(\partial^0_1\partial^0_2x = \chi(i,0)\). Indeed, in this case, \(\partial^1_1\partial^0_1x \notin (P_\C)_0\) or \(\partial^1_1\partial^0_2x \notin (P_\C)_0\), and so \(\partial^0_1\partial^0_2x\) is a vertex of \(\C\) in which starts an edge that ends in a vertex of \((P_\B)_0\setminus (P_\C)_0\). The only such vertex is \(\chi(i,0)\). Thus, there exists \(j \in \{i, \dots, k\}\) such that \(\partial^0_1\partial^0_2x = \chi(j,0)\).
	
	Suppose that \(\partial^0_1\partial^0_2x = \chi(j,0)\) with \(i < j \leq k\). Since \(\lambda_\Q(\partial^0_2x) <  \lambda_\Q(\partial^0_1x)\), we have \(\partial^0_1x \not= \partial^0_2x\). Therefore there exists \(r \in \{1,2\}\) such that \(\partial^0_rx = \chi(j,[0,1])\) and \( \partial^0_{3-r}x = \chi([j,j+1],0)\). Since \(\lambda_\Q(\partial^1_rx) = \lambda_\Q(\partial^0_rx) = a\), we have \(\partial^1_rx = \chi(j+1,[0,1])\). Since \(\Q_{\leq 1} = \B_{\leq 1}\) is deterministic, we have \(\partial^1_{3-r}x = \chi([j,j+1], 1)\) because \(\partial^1_{3-r}x\) starts in \(\chi(j,1) = \partial^0_1\chi([j,j+1], 1)\) and \(\lambda_\Q(\partial^1_{3-r}x) = \lambda_\Q(\partial^0_{3-r}x) = \lambda_\Q(\chi([j,j+1],0)) = \lambda_\Q(\chi([j,j+1],1))\).
	Since \(\Q\) satisfies HM3 and \(<\) is asymmetric, it follows that \(x = \chi([j,j+1],[0,1]) \in (P_\B)_2\). 
	
	Suppose now that \(\partial^0_1\partial^0_2x = \chi(i,0)\). Since \(\lambda_\Q(\partial^0_2x)< \lambda_\Q(\partial^0_1x)\), we have \(\partial^0_1x \not= \partial^0_2x\). The edges starting at \(\chi(i,0)\) are \(\chi(i,[0,1])\) and \({\chi([i,i+1],0)}\), and  \(\xi(i-1, [1,2])\) when \(i > 0\). Since all edges starting at the endpoints of \(\chi(i,[0,1])\) and \(\xi(i-1, [1,2])\) are edges of \(\C\), there exists \(r \in \{1,2\}\) such that \(\partial^0_rx\) is \(\chi(i,[0,1])\) or \(\xi(i-1, [1,2])\}\), and \( \partial^0_{3-r}x = \chi([i,i+1],0)\). Since all edges starting at the endpoints of \(\chi(i,[0,1])\) and \(\xi(i-1, [1,2])\) end in reachable vertices of \(\C\) and there exists no edge starting in a reachable vertex of \(\C\) and ending in \(\chi(i+2, 0)\), we have \(\partial^1_1\partial^1_2x = \chi(i+1,1)\), \(\partial^1_rx= \chi(i+1,[0,1])\), \(\partial^0_rx= \chi(i,[0,1])\), and \(\partial^1_{3-r}x = \chi([i,i+1],1)\). Since \(\Q\) satisfies HM3 and \(< \) is asymmetric, it follows that \(x = \chi([i,i+1],[0,1]) \in (P_\B)_2\). 
	
	Suppose that there exists an element \(y \in (P_\Q)_3\) such that at least one face of \(y\), say \(\partial^l_jy\), does not belong to \(\C\). Then \(\partial^l_jy = \chi([r,r+1],[0,1])\) for some \(i \leq r \leq k\). Since, for some \(s \in \{1,2\}\), 
	\begin{align*}
		\chi([r,r+1],0) &= \partial^0_s\chi([r,r+1],[0,1]) = \partial^0_s\partial^l_jy = \begin{cases}
			\partial^l_{j-1}\partial^0_sy, & s <j,\\
			\partial^l_j\partial^0_{s+1}y, & s \geq j,
		\end{cases}		
	\end{align*}
	Lemma \ref{facelem} implies that \((\chi([r,r+1],0)\) is an edge of two distinct faces of \(y\). Since \(\chi([r,r+1],[0,1])\) is the only 2-cube of \(\B\) having \((\chi([r,r+1],0)\) as an edge, this is impossible. 
	
	A simple induction now shows that \((P_\Q)_m = (P_\C)_m\) for all \(m \geq 3\). It follows that \(\B\) satisfies HM4. 
\end{proof}

\section{Shared-variable systems} \label{SecPG}

In this section, we consider shared-variable systems given by program graphs and establish our main result:

\begin{theor} \label{mainresult}
	There exists a shared-variable system such that the geometric realization of its HDA model has the homotopy type of the polyhedron \(|K|\).
\end{theor}

\subsection*{Program graphs and shared-variable systems} 

Let $V$ be a set of \emph{variables}. The \emph{domain} of a variable $x$, i.e., the set of its possible values, will be denoted by $D_x$. A \emph{program graph} over $V$ is a tuple \[(L, A, T, g, \imath
)\]
where $L$ is a set of \emph{locations} or \emph{local states}, $A$ is a finite set of \emph{actions}, i.e., functions ${\prod \limits_{x \in V} D_x \to \prod \limits_{x \in V} D_x}$, $T\subseteq L \times A  \times L$ is a set of \emph{transitions}, \(g\) is a function that specifies a \emph{guard condition}, i.e., a subset of  \(\prod \limits_{x \in V} D_x\), for each transition, and $\imath \in L$ is an \emph{initial location} (cf. \cite{BaierKatoen}). A \emph{shared-variable system} over $V$ is a tuple $(\G_1, \dots, \G_n, \eta)$ consisting of program graphs \(\G_i\) and an \emph{initial evaluation} \(\eta \in \prod \limits_{x \in V} D_x\).

\subsection*{The HDA model of a shared-variable system}

Consider a shared-variable system $(\G_1, \dots, \G_n, \eta)$ over a set of variables $V$, and write  
\(\G_i = (L_i, A_i, T_i, g_i, \imath_i)\). The \emph{state graph} of $(\G_1, \dots, \G_n, \eta)$ is the 1-truncated   precubical set \(Q\) where 
\[Q_0 = L_1 \times \cdots \times L_n \times \prod \limits_{x \in V} D_x,\] 
\[Q_1 = \bigcup_{\substack{i\in \{1,\dots,n\} \\t \in T_i}} L_1 \times \cdots \times  L_{i-1} \times \{t\} \times L_{i+1} \times \cdots \times L_n \times g_i({t}),\]
and for \(y = (l_1,\dots, l_{i-1},t,l_{i+1}, \dots, l_{n},\gamma)\in Q_1\) with \(t = (l^0_t, a_t, l^1_t)\), \[d^0_1y = (l_1,\dots, l_{i-1},l^0_{t},l_{i+1}, \dots, l_{n},\gamma)\] and  \[d^1_1y = (l_1,\dots, l_{i-1},l^1_{t},l_{i+1}, \dots, l_{n},a_{t}(\gamma)).\]
The \emph{initial state} of the system is the state
$I  = (\imath_1, \dots , \imath_n, \eta)$. 
The \emph{transition system model} of  $(\G_1, \dots, \G_n, \eta)$ is the transition system \(\T\) where $P_\T$ is the largest precubical subset of $Q$ such that all states are reachable from the initial state \(I\), $I_\T = I$, $\Sigma_\T = \bigcup \limits_{i= 1}^n \{i\}\times A_i$, and the label of an edge $y = (l_1,\dots, l_{i-1},t,l_{i+1}, \dots, l_{n},\gamma) \in (P_\T)_1$ with \(t = (l^0_t, a_t, l^1_t)\) is given by ${\lambda_\T(y) = (i,a_t)}$. The \emph{HDA model} of \((\G_1, \dots, \G_n, \eta)\) is the HDA model of \(\T\) with respect to the relation \(R\) on \(\Sigma_\T\) given by \[(i,a) \,R\, (j,b) \iff i < j.\]

\begin{rem}
	In practice, the transition system model of a shared-variable system can be constructed without handling unreachable states using a procedure such as the one described in the reference manual of the Spin model checker \cite[Sect. 7]{Spin}. HDA models of shared-variable systems written in Promela, the process description language of Spin, can be computed using the tool pg2hda \cite{pg2hda}.
\end{rem}

\subsection*{Proof of Theorem \ref{mainresult}}

Since \(K\) is connected and \(|K| \approx |P| = |P_\P|\) by Theorem \ref{homeo}, the HDA \(\P\) is also connected. By Proposition \ref{Pbidet} and Theorems \ref{PHDAmod} and  \ref{thmacc}, there exists an accessible and deterministic HDA \(\B\) which is the HDA model of its 1-skeleton with respect to a total order on \(\Sigma_\B\) and satisfies \(|P_\B| \simeq |P_\P| \approx |K|\). 

Suppose that \(\Sigma_\B = \{a_1 < \dots < a_n\}\). Consider a single variable \(x\) with domain \(D_x = (P_\B)_0\), and let \((\G_1, \dots, \G_n, \eta)\) be the shared-variable system over \(V = \{x\}\) where \(\eta = I_\B\) and the program graphs \({\G_i = (L_i,A_i, T_i, g_i, \imath_i)}\) are defined by
\begin{itemize}
    \item \(L_i = \{0\}\);
    \item \(A_i = \{\bar a_i\}\) where \(\bar a_i(v) = \left\{\begin{array}{ll}
        d^1_1y, & \exists\, y \in (P_\B)_1 : d^0_1y = v,\,  \lambda_\B(y) = a_i,\\
        v, & \mbox{else;}
    \end{array}\right.\)
    \item \(T_i = \{(0,\bar a_i,0)\}\);
    \item \(g_i(0,\bar a_i,0) = \{d^0_1y\,|\, y \in (P_\B)_1,\, \lambda_\B(y) = a_i\}\);
    \item \(\imath_i = 0\).		
\end{itemize}
Since \(\B\) is deterministic, the action \(\bar a_i\) is well defined. Let \(Q\) be the state graph of \((\G_1, \dots, \G_n, \eta)\). We have \[(P_\B)_0 \cong L_1 \times \dots \times L_n \times D_x = Q_0.\]
Since \(\B\) is deterministic, the map 
\begin{align*}
    d^0_1\colon \{y \in (P_\B)_1\,|\, \lambda_\B(y) = a_i\} &\to \{d^0_1y\,|\, y \in (P_\B)_1,\, \lambda_\B(y) = a_i\} = g_i(0,\bar a_i,0)
\end{align*}
is a bijection for each \(i\). Hence	
\begin{align*}
    (P_\B)_1 &= \bigcup_{i\in \{1,\dots,n\}} \{y \in (P_\B)_1\,|\, \lambda_\B(y) = a_i\}\\ 
    &\cong  \bigcup_{i\in \{1,\dots,n\}} L_1 \times \cdots \times  L_{i-1} \times \{(0,\bar a_i,0)\} \times L_{i+1} \times \cdots \times L_n \times g_i(0,\bar a_i, 0)\\		
    &= Q_1.
\end{align*}
Since for \(y\in (P_\B)_1\) with \(\lambda_\B(y) = a_i\) we have 
\[d^0_1(0, \dots, 0, (0,\bar a_i,0), 0, \dots, 0, d^0_1y) = (0, \dots, 0, d^0_1y) \]
and 
\[d^1_1(0, \dots, 0, (0,\bar a_i,0), 0, \dots, 0, d^0_1y) = (0, \dots, 0, \bar a_i(d^0_1y)) = (0, \dots, 0, d^1_1y),\]
the precubical sets \((P_\B)_{\leq 1}\) and \(Q\) are isomorphic.

Let \(\T\) be the transition system model of \((\G_1, \dots, \G_n, \eta)\). Since the initial state of the system, \(I = (0, \dots, 0, I_\B)\), corresponds to \(I_\B\) under the isomorphism \((P_\B)_{\leq 1} \cong Q\) and \(\B\) is accessible, all states of \(Q\) are reachable from \(I\). Hence \(P_\T = Q\). We have 
\[\Sigma_\B = \{a_1, \dots, a_n\} \cong \{(1, \bar a_1), \dots , (n,\bar a_n)\} = \Sigma_\T.\] 
Since for an edge \(y \in (P_\B)_1\) with \(\lambda_\B(y) = a_i\) we have 
\[\lambda_\T(0,\dots, 0, (0, \bar a_i, 0),0, \dots, 0, d^0_1y) = (i, \bar a_i),\] 
it follows that the transition systems \(\B_{\leq 1}\) and \(\T\) are isomorphic.

Let \(\A\) be the HDA model of \((\G_1, \dots, \G_n, \eta)\). Then \(\A\) is the HDA model of \(\T\) with respect to the relation 
\(R\) on \(\Sigma_\T\) given by \[(i,\bar a_i) \,R\, (j,\bar a_j) \iff i < j \iff a_i < a_j.\]
By \cite[Thm. 4.2, Cor. 4.5]{transhda}, it follows that the HDAs \(\A\) and \(\B\) are isomorphic. In particular, \(|P_\A| \approx |P_\B| \simeq |K|\). \qed{\parfillskip0pt\par}

\bibliography{refs}
\bibliographystyle{apalike}

\end{document}